\documentclass[11pt,onecolumn]{article}
\usepackage{amsmath,amsthm,amssymb,amsfonts,epsfig,graphicx}
\usepackage{hyperref}
\usepackage{cleveref}
\usepackage{color}
\usepackage{enumitem}
\usepackage{url}
\usepackage{array,multirow}
\usepackage{arydshln}
\usepackage{bm}
\usepackage{todonotes}
\usepackage{tikz}
\usetikzlibrary{arrows.meta}

\usepackage{stmaryrd}
\usepackage[T1,OT1]{fontenc}
\DeclareTextCommandDefault{\k}[1]{{\fontencoding{T1}\selectfont\k{#1}}}%

\usepackage{enumitem}
\usepackage[vlined,linesnumbered,ruled,resetcount]{algorithm2e}

\usepackage[left=1.0in,top=0.5in,right=1.05in,bottom=1.0in,nohead,textheight=10in,footskip=0.3in]{geometry}

\newtheorem{theorem}{Theorem}
\newtheorem{corollary}[theorem]{Corollary}

\newtheorem{assumption}{Assumption}
\newtheorem{lemma}[theorem]{Lemma}
\newtheorem{definition}[theorem]{Definition}
\newtheorem{proposition}[theorem]{Proposition}
\newtheorem{remark}{Remark}

\SetKwInput{kwMyRepeat}{Repeat}

\def\myparagraph#1{\vspace{2pt}\noindent{\bf #1~~}}

\crefname{theorem}{Theorem}{Theorems}
\crefname{algocf}{Algorithm}{Algorithms}
\crefname{problem}{problem}{problems}
\Crefname{problem}{Problem}{Problems}
\crefformat{problem}{problem~(#2#1#3)}
\Crefformat{problem}{Problem~(#2#1#3)}

\newcommand{\RR}{\mathbb{R}}

\newcommand{\defn}[1]{\textbf{#1}}

\newcommand{\tflow}{T_{{\rm maxflow}}}
\newcommand{\Scut}{S_{\text{cut}}}
\newcommand{\Tcut}{T_{\text{cut}}}
\newcommand{\OPT}{{\rm OPT}}
\newcommand{\midweight}{w}

\long\def\ignore#1{}
\def\myps[#1]#2{\includegraphics[#1]{#2}}

\def\br(#1,#2){{\langle #1,#2 \rangle}}
\def\setZ[#1,#2]{{[ #1 .. #2 ]}}
\def\Pr#1{\mbox{\tt Pr}\left[{#1}\right]}

\newcommand{\eqdef}{{\stackrel{\mbox{\tiny \tt ~def~}}{=}}}

\def\q={\quad=\quad}
\def\qq={\qquad=\qquad}

\def\calD{{\cal D}}

\def\calK{{\cal K}}
\def\calG{{\cal G}}
\def\calI{{\cal I}}

\def\calL{{\cal L}}
\def\calN{{\cal N}}

\def\calP{{\cal P}}

\def\calS{{\cal S}}

\def\calQ{{\cal Q}}

\def\XX(#1){{{#1}^\downarrow}}

\def\psfile[#1]#2{}
\def\psfilehere[#1]#2{}

\def\assign(#1,#2){\langle#1,#2\rangle}
\def\edge(#1,#2){(#1,#2)}

\def\slack(#1){\texttt{slack}({#1})}
\def\barslack(#1){\overline{\texttt{slack}}({#1})}

\def\unitvec(#1){{{\bf u}_{#1}}}

\newcommand{\newstuff}[1]{ #1}

\def\br#1{{\llbracket #1 \rrbracket}}

\title{\Large\bf  }
\author{}
\title{
\Large\bf  \vspace{-20pt} Fast Approximation Algorithms for \\ Balanced Edge and Vertex Separators
}
\author{Vladimir Kolmogorov \\ \normalsize Institute of Science and Technology Austria (ISTA) \\ {\normalsize\tt vnk@ist.ac.at}
\and
Jack Spalding-Jamieson \\ \normalsize Independent \\ {\normalsize\tt jacksj@uwaterloo.ca}}
\date{}

\begin{document}
\maketitle


%\todo[inline]{J: I've become aware of a line of work on vertex-capacitated max flow that predates the directed one (\url{https://doi.org/10.1109/FOCS52979.2021.00100}). I don't think we should modify the algorithm anywhere, but it would be a good idea to cite this in a footnote in a new version to send to ESA (and potentially update on arxiv).}

\begin{abstract}
%{\color{red} [I changed ``$\Theta(c)$-balanced'' to
%``$\Theta(1)$-balanced'', and $\Theta(\varepsilon)$ to $O(\varepsilon)$.]} % sounds good
We present a family of fast pseudo-approximation algorithms for
the minimum balanced element cut problem in a graph, which includes vertex and edge separators.
Given a graph $G=(V,E)$ with polynomially-bounded finite integer element weights
with $n$ vertices and $m$ edges,
and a (constant) balance parameter $c\in(0,1/2)$,
where $G$
has some (unknown) $c$-balanced element cut of size $\OPT_c$,
we give a
(Monte-Carlo randomized) algorithm running in $O(n^{\varepsilon}m^{1+o(1)})$ time
that produces a $\Theta(1)$-balanced element cut of size $O(\OPT_c\cdot\sqrt{(\log n)/\varepsilon})$
for
any value $\varepsilon\in[\Theta(1/\log(n)),\Theta(1)]$.
In particular,
for any function $f(n)=\omega(1)$
(including $f(n)=\log\log n$, for instance),
%with $\varepsilon=\Theta(1/\log\log n)$,
we can
produce an element cut of size $O(\OPT_c\cdot\sqrt{\log n}\cdot f(n))$
in time $O(m^{1+o(1)})$.
Moreover, for an arbitrarily small constant $\varepsilon=\Theta(1)$,
our algorithm also achieves the best-known approximation ratio for the vertex-separator specialization
in $O(m^{1+\varepsilon})$ time.
%but with a significantly faster runtime than even much weaker previously-known approximation algorithms.

Our algorithms are
based on a semidefinite programming (SDP) relaxation of the problem.
A standard approach to solving an SDP relaxation for this kind of problem would use
the primal-dual Matrix Multiplicative Weights Update framework of Arora and Kale,
but this framework does not seem to apply,
since we cannot express all of our SDPs primal variables explicitly.
Instead, we present a new MMWU framework for our problem
that allows us to represent some of our variables implicitly.
We also extend our framework to a general \textbf{Partial-variable MMWU framework},
which we expect will be useful for future work.
Our oracle for the balanced element cut problem within our framework uses $O(n^{\varepsilon}\text{polylog}(n))$ almost-linear time maximum-flow computations,
and would be sped up if the time complexity of maximum-flow improves.

For graphs with {\em padding bound} $\beta$, another oracle gives an
$O(\sqrt\beta/\varepsilon)$ approximation in $m^{1+o(1)}n^{\varepsilon}\operatorname{polylog}(n)$ time.
It is known that $\beta\le O(\log h)$ for $K_h$-minor-free graphs (and in particular $\beta\le O(\log n)$).
\end{abstract}

\section{Introduction}
\label{sec:intro}

Partitioning a graph into smaller pieces is a fundamental problem
that arises in many fields.
One kind of graph partition with widespread applications in
algorithms is the \defn{vertex separator}:
A small set of vertices whose removal disconnects the graph into small components.

More generally, let $G=(V,E)$ be a graph with $n=|V|$ vertices,
let $w:V\cup E\to\mathbb Z_{\ge0}$ assign a polynomially-bounded weight to each vertex and edge,
and let $c\in(0,1/2)$ be a constant.
For $Z\subseteq V\cup E$, write $w(Z)=\sum_{z\in Z}w(z)$.
We treat vertex and edge separators in a common framework.
A \defn{$c$-balanced element cut} is a partition $V=A\sqcup B\sqcup C$ with
$\max\{|A|,|B|\}\le(1-c)n$, of cost
\[
 \operatorname{cost}(A,B,C)=w(C)+w(E(A,B)).
\]
It deletes the vertices in $C$ and the edges between $A$ and $B$.
We also use $w(z)=\infty$ as notation for forbidding deletion of an element $z$.
Vertex separators are the specialization $w(e)=\infty$ for all $e\in E$, and edge cuts are the specialization $w(x)=\infty$ for all $x\in V$.
Write $\OPT_c$ for the minimum finite cost of a $c$-balanced element cut,
$w_{\max}$ for the maximum finite element cost, and $W$ for the sum of all finite element costs.
Balance always counts original vertices.
For fixed sets to separate, this is weighted element connectivity: attach an undeletable terminal leaf to each original vertex by an undeletable edge, and regard the original vertices as deletable nonterminals.
Thus paths pay for their original endpoints as well as their internal vertices.
%The standard element-connectivity definition and its max-flow reduction appear in
%\href{https://arxiv.org/html/2201.00408v2#S2}{Pettie, Saranurak and Yin, \emph{Space Complexity of Vertex Connectivity Oracles}, Section~2.1 and footnote~6 in the proof of Lemma~2.6}.

Small vertex separators are known to exist for many classes of graphs
with uniform weights $\midweight(\cdot)=1$.
Planar graphs are known to admit $\frac13$-balanced vertex separators of size $O(\sqrt{n})$~\cite{ungar1951theorem,lipton1979separator},
genus-$g$ graphs are known to admit $\frac13$-balanced vertex separators of size $O(\sqrt{gn})$~\cite{gilbert1984separator},
and similar results are known for minor-free graphs~\cite{kawarabayashi2010separator}, geometric graphs~\cite{miller98,spalding2025reweighted}, induced-minor-free graphs~\cite{korhonen2024},
and other classes~\cite{matousek2014,davies2025}.
These small separators also have numerous applications.
To name a few, they naturally allow devising divide and conquer algorithms,
dynamic programming algorithms,
subexponential algorithms for NP-hard problems on special graphs,
and maximum matching algorithms~\cite{separatorsbook,lipton1980,korhonen2024}.

\newstuff{
Unfortunately,
while optimal planar separators can be found in linear time,
none of these other results correspond to particularly fast algorithms
%(fast ideally meaning close-to-linear),
often requiring exponential time
%(or worse!)
in a parameter.
For genus-$g$ graphs, the best-known separator theorem
requires a surface-map/rotation system~\cite{gilbert1984separator},
which takes time exponential in $g$ to compute~\cite{mohar1999linear,brinkmann2020practical},
meaning it cannot be used to compute separators
for graphs with even a small genus of ${\rm polylog}(n)$ in polynomial time.
In the case of minor-free graphs, the time complexity can be even worse:
One of the best-known separator theorems~\cite{kawarabayashi2010separator} corresponds
to an algorithm which uses time proportional to multiple tower functions
in the size of the minor.
At best, some of the constructions correspond to polynomial-time algorithms (with a large polynomial overhead),
such as those of induced-minor-free graphs~\cite{korhonen2024},
but such results are typically conjectured to be significantly sub-optimal existential bounds for the classes to begin with.
Due to the significant overhead of these algorithms,
many direct applications of recursively-separable graphs
assume some routine is provided that produces a separator,
and the time complexity of the assumed routine is incorporated into
their final result~\cite{dong2024fast,
lahn2021weighted,
alon2013matrix,
ducoffe2022diameter,
kerber2016persistent,
elberfeld2025space,
cheung2013graph,
lipton1979generalized,
zou2026generalized}.

Even for cases where the graph class might
admit a fast algorithm (such as planar graphs),
such algorithms only apply to graphs within their respective classes.
In most practical settings for graph partitioning,
we do not have any special guarantees that the graph falls into a certain class,
but a graph might still have a small separator without properly falling into a special class:
For instance, road networks, which are known to admit \emph{empirically} good separators~\cite{delling2011graph},
are \emph{almost} planar,
in the sense that removing a relatively small number of vertices corresponding
to structures like bridges will make them planar.
Or, more theoretically, a graph could be nearly minor-free, in the sense that it contains
a very large minor-free induced subgraph.
}
%Moreover, these existential results also have two more faults.
%First, none of these results capture robustness to graphs outside of their class:
%If a graph $G$ were made up of a genus-$g$ subgraph $H$, and some small number $c$ of extra vertices,
%then we know that it contains a separator of size at most $O(\sqrt{gn})+c$.
%However, we have no way of recovering such a separator

%Although many classes of graphs are known to admit small separators,
%the corresponding algorithms for constructing them are often quite slow,
%often taking exponential time, or requiring some special information to be known about the graph.
%For instance,
Thus, we naturally turn to the question of approximation algorithms:
Given a graph $G$, can we approximate the minimum $c$-balanced vertex separator?
It turns out that this is essentially impossible as-stated:
It is NP-hard to approximate the minimum $c$-balanced vertex separator up to an \emph{additive} approximation
of $O(n^{\frac12-\epsilon})$ for any $\epsilon>0$~\cite{bui1992finding}.
Fortunately, we may instead use \emph{pseudo}-approximations (also called bi-criteria approximations):
For a graph $G$ whose minimum $c$-balanced vertex separator has size $\OPT_c$,
we can produce a $\Theta(c)$-balanced vertex separator with size close to $\OPT_c$.

Our first result
%, which subsumes all those previously-mentioned,
is the following fast pseudo-approximation algorithm for finding a minimum balanced element cut.
(When stating complexities, we always assume that $m=\Omega(n)$ to simplify expressions;
if it is not the case, then $m$ would need to be replaced with $m+n$).

%although such a thing seems thought to be possible with some existing methods.

\begin{theorem}
\label{thm:main}
For any $\varepsilon\in[\Theta(1/\log n),\Theta(1)]$
and any constant $c\in(0,1/2)$,
%with $c=\Theta(1)$,
there is a Monte-Carlo randomized algorithm (succeeding w.h.p.) that,
given a graph $G=(V,E)$
and polynomially-bounded finite integer element weights $w$,
finds a $\Theta(1)$-balanced element cut of size at most
$O(\OPT_c\cdot\sqrt{(\log n)/\varepsilon})$,
where
%$c'=\Theta(c)$ and 
$\OPT_c$ is the minimum size of a $c$-balanced element cut.
The algorithm runs in time
%$$
%O\left(mn^{O(\varepsilon)}{\rm polylog}(nw_{\max})+n^{O(\varepsilon)}\cdot\tflow(n+m,n+m)\log(nw_{\max})\right),
%$$
$O\left(n^{\varepsilon}\cdot\tflow(m,m)\cdot{\rm polylog}(n)\right)$,
where $\tflow(N,M)$ is the time to solve a max-flow problem on a graph with $N$ vertices and $M$ edges
with capacities all ${\rm poly}(n)$.
\end{theorem}

By applying the almost-linear time maximum flow algorithm
of Chen, Kyng, Liu, Peng, Gutenberg, and Sachdeva~\cite{linear-maxflow},
we know that $\tflow=O(m^{1+o(1)})$.
%
%Consequently, the following immediate corollary to \cref{thm:main} (using $\varepsilon=\Theta(1/f(n))$
%for a slow-growing function $f(n)$)
%is a significant step forward:
%
%The following is an immediate corollary to \cref{thm:main},
%using $\varepsilon=\Theta(1/f(n))$
%for a slow-growing function $f(n)$:
%
%\begin{corollary}
%Let $f(n)=\omega(1)$ be a function (e.g., $f(n)=\log\log n$).
%There is a Monte-Carlo randomized algorithm (succeeding w.h.p.) that,
%given a graph $G=(V,E)$
%and polynomially-bounded finite integer element weights $w$,
%finds a $\Theta(1)$-balanced element cut of size at most
%$O(\OPT_c\cdot\sqrt{\log n}\cdot f(n))$.
%The algorithm runs in $O(m^{1+o(1)})$ time,
%plus the time to evaluate $f(n)$.
%\end{corollary}
Specific choices of $\varepsilon$ lead to the following conclusions.
\begin{itemize}
\item Taking any fixed $\varepsilon>0$, we obtain $O(\sqrt{\log n})$-approximation in $O(m^{1+\varepsilon})$ time. This approximation factor is the best known for both vertex- and edge balanced separators.
\item Taking $\varepsilon=\Theta(1/\log n)$,
we obtain an $O(\log n)$-approximation in 
$O(\tflow(m,m)\cdot{\rm polylog}(n))$ time.
This will be useful if $\tflow(m,m)$ is ever improved to be near-linear.
\end{itemize}

%On the other hand, the best-known approximation factor
%for the vertex-separator specialization is $O(\sqrt{\log n})$~\cite{FeigeHL08},
%and we can match that result with an algorithm
%that gets arbitrarily close to almost-linear time:
%
%\todo[inline]{These can be quick items instead of full corollaries}
%
%\begin{corollary}
%For any constant $\delta>0$,
%there is a Monte-Carlo randomized algorithm (succeeding w.h.p.) that,
%given a graph $G=(V,E)$
%and polynomially-bounded finite integer element weights $w$,
%finds a $\Theta(1)$-balanced element cut of size at most
%$O_\delta(\sqrt{\log n}\cdot\OPT_c)$.
%The algorithm runs in $O(m^{1+\delta})$ time.
%\end{corollary}
%
%Finally, by choosing $\varepsilon=\Theta(1/\log n)$ in
%\cref{thm:main},
%we can attain a time complexity that is entirely dependent on $\tflow(m,m)$:
%\begin{corollary}
%There is a Monte-Carlo randomized algorithm (succeeding w.h.p.) that,
%given a graph $G=(V,E)$
%and polynomially-bounded finite integer element weights $w$,
%finds a $\Theta(1)$-balanced element cut of size at most
%$O(\OPT_c\cdot\log n)$.
%The algorithm runs in $O(\tflow(m,m)\cdot{\rm polylog}(n))$ time.
%\end{corollary}
%This last corollary will be useful if $\tflow(m,m)$ is ever improved to be near-linear.

\subsection{Padding bound}

\begin{definition}\label{def:rect-padding}
Graph $(V,E)$ has padding bound $\beta\ge1$ if, for every nonnegative edge length vector $\xi\in\RR_{\ge0}^E$ and every $D>0$, the shortest-path pseudometric $d(\cdot,\cdot)$
with edge lengths $\xi_{xy}$ admits a distribution of partitions into sets of diameter at most $D$ such that
$$
 \Pr{B(x,\Gamma D)\subseteq P(x)}\ge e^{-\Gamma\beta}\qquad(x\in V)
$$
where $B(x,r)=\{y\in V\::\:d(y,x)\le r\}$ is the ball of radius $r$ centered at $x$,
and $P(x)$ is the part containing $x$. Here $\Gamma>0$ is a fixed absolute constant.
We write $\beta(G)$ for the smallest padding bound of $G=(V,E)$.

If we restrict to the case of edge length vectors $\xi\in\RR_{\ge0}^E$
arising from vertex lengths $\xi_v\in\RR_{\ge0}^V$
(so $\xi_{uv}=\frac12\xi_u+\frac12\xi_v$),
we obtain an analogous quantity
we call $\beta_v(G)$.
\end{definition}

Many existing results are known for these padding bounds. Most are implied by the following results.
\begin{theorem}[Conroy--Filtser~\cite{ConroyFiltser2025}, Theorem~2]\label{thm:minor-free-padding}
Every $K_h$-minor-free graph $G$ has $\beta(G)\in O(\log h)$.
\end{theorem}

\begin{theorem}[{\cite[Theorem~1.4]{DRSW27}}]\label{thm:string-padding}
Every string graph $G$ has $\beta_v(G)\in O(1)$.
\end{theorem}

We will prove the following algorithmic result in terms of the padding bound.
\begin{theorem}
\label{thm:main:padding}
For any $\varepsilon\in[\Theta(1/\log n),\Theta(1)]$
and any constant $c\in(0,1/2)$,
there is a Monte-Carlo randomized algorithm (succeeding w.h.p.) that,
given a graph $G=(V,E)$ with padding bound $\beta$
(or with $\beta_v(G)\le\beta$ in the special case of vertex separators)
and polynomially-bounded finite integer element weights $w$,
finds a $\Theta(1)$-balanced element cut of size at most
$O(\OPT_c\cdot\sqrt{\beta}/\varepsilon)$,
where
$\OPT_c$ is the minimum size of a $c$-balanced element cut.
The algorithm runs in time
$m^{1+o(1)}n^{\varepsilon}\operatorname{polylog}(n)$.
For fixed $\varepsilon>0$, this is $O(m^{1+\varepsilon})$.
\end{theorem}

In particular, by taking any fixed $\varepsilon>0$ we obtain an $O(\sqrt{\beta})$ approximation in $O(m^{1+\varepsilon})$ time.

%Taking any fixed $\varepsilon>0$ gives the following consequence.
%\begin{corollary}\label{cor:element-padding}
%For every fixed $c\in(0,1/2)$ and every constant $\delta>0$, there is a randomized algorithm running in $O(m^{1+\delta})$ time which, given a graph $G$ with padding bound $\beta$ (or with $\beta_v(G)\le\beta$ in the special case of vertex separators) and polynomially-bounded finite integer element weights $w$, returns with high probability a $c'$-balanced element cut of cost
%\[
% O_\delta(\sqrt\beta)\,\OPT_c,
%\]
%where $c'>0$ depends only on $c$.
%\end{corollary}
For comparison, when the partitions in \cref{def:rect-padding} can be sampled efficiently,
a standard polynomial-time LP-based approach~\cite{leighton1999multi} gives a $O(\beta)$-approximation
(e.g., a $O(\log h)$-approximation for $K_h$-minor-free graphs~\cite{ConroyFiltser2025}).
\Cref{thm:main:padding} improves this to $O(\sqrt\beta)$.
Moreover, it does not require that the partitions can be sampled efficiently.
See \cref{tab:classes} for the direct applications to various graph classes.

\begin{table}[htpb]
    \centering
    \footnotesize
    \setlength{\tabcolsep}{3.5pt}
    \renewcommand{\arraystretch}{1.15}
    \begin{tabular}{|l|l|l|c|l|l|}
    \hline
    \textbf{Class} & \textbf{Cut} & \textbf{Previous ratio} & \textbf{Year} & \textbf{Ref} & \textbf{This paper} \\
    \hline
    $K_h$-minor-free & edge, vertex & $O(\log h)$ & 2025 & \cite{ConroyFiltser2025} & $O(\sqrt{\log h})$ \\
    \hline
    Genus $g$ & edge, vertex & $O(\log g)$       \footnotemark 
    & 2010 & \cite{lee2010genus} & $O(\sqrt{\log g})$ \\
    \hline
    \multirow{2}{*}{Treewidth $\mathrm{tw}$}
        & vertex & $O(\sqrt{\log\mathrm{tw}})$    %%%%% TODO: restore \footnotemark 
        & 2005 & \cite{FeigeHL08} & \multirow{2}{*}{$O(\sqrt{\log\mathrm{tw}})$} \\
        & edge & $O(\log\mathrm{tw})$ & 2024 & \cite{filtser2025optimal} & \\
    \hline
    String graphs & vertex & $O(1)$ & 2026 & \cite{DRSW27} & $O(1)$ \\
    \hline
    \end{tabular}
    \caption{Class-dependent ratios of polynomial-time pseudo-approximations for balanced separators.
    %Ours are achieved in $O(m^{1+\delta})$ time for any constant $\delta>0$ (\cref{cor:element-padding}).
    }
    \label{tab:classes}
\end{table}
%%%%%%%%%%%%%%%%%%%%%%%%%%%%%%%%%%%%%%%%%%%%%%%%%%%%%%%%%%%%
%\addtocounter{footnote}{-1}%
\footnotetext{The history of this bound is nuanced. A $O(\sqrt{\log g})$-approximation was claimed in~\cite{lee2010genus}, but a proof was absent, and one of the explicit claimed theorems is false~\cite[Appendix~B]{spalding2025reweighted}.
In the same work, a padding bound of $O(\log g)$ was claimed, but this bound also lacks proof.
Ratio $O(\log g)$ also follows from a later work that does include a proof~\cite{ConroyFiltser2025}.}
%
%\stepcounter{footnote}\footnotetext{The history of this bound is nuanced.
%The original proof used an incorrect argument, but we have been informed of an %upcoming fix~\cite{FHLfix}.\todo[inline]{Need permission to write this.}}
%%%%%%%%%%%%%%%%%%%%%%%%%%%%%%%%%%%%%%%%%%%%%%%%%%%%%%%%%%%%

\subsection{History}

\begin{table}[htpb]
    \centering
    \footnotesize
    \setlength{\tabcolsep}{3.5pt}
    \setlength{\dashlinedash}{1.5pt}
    \setlength{\dashlinegap}{2pt}
    \renewcommand{\arraystretch}{1.15}
    \begin{tabular}{|>{\raggedright\arraybackslash}m{0.20\textwidth}>{\raggedright\arraybackslash}m{0.20\textwidth}>{\centering\arraybackslash}m{0.075\textwidth}>{\centering\arraybackslash}m{0.10\textwidth}>{\raggedright\arraybackslash}m{0.32\textwidth}|}
    \hline
    \textbf{Separator Size} & \textbf{Running Time} & \textbf{Year} & \textbf{Ref} & \textbf{Notes} \\
    \hline
    $O(\OPT\cdot\log n)$
        & $O(n)$ large LP solves
        & 1988
        & \cite{leighton1988approximate,leighton1999multi}
        & \\
    \hdashline
    $O(\OPT\cdot\sqrt{\log n})$
        & $O(n)$ large SDP solves
        & 2005
        & \cite{FeigeHL08,agarwal2005}
        & Best-known polynomial-time approximation ratio. \\
    \hdashline
    $O(\OPT\cdot\log n)$
        & $\widetilde{O}(m^{1.5})$
        & 2007
        & \cite{AKstoc,AK}
        & Uses $\widetilde{O}(1)$ strongly-polynomial max-flows. \\
    \hdashline
    $O(\OPT\cdot\sqrt{\log n})$
        & $\widetilde{O}(n^2+m^{1.5})$
        & 2007
        & \cite{AKstoc,AK}
        & \\
    \hdashline
    $\widetilde{O}(\OPT)$
        & $\widetilde{O}(m^{1+o(1)})$
        & 2010
        & \cite{louis2010cutmatchinggamesdirectedgraphs}
        & \textbf{Proof missing.} \\
    \hdashline
    $\widetilde{O}(\OPT^2)$
        & $\widetilde{O}(\OPT^3 m)$
        & 2019
        & \cite{brandt2019approximating}
        & \\
        %& Runtime and guarantee depend on $\OPT$. \\
    \hdashline
    $O(\OPT\sqrt{\log n/\varepsilon})$
        & $O(n^{1+O(\varepsilon)}m^{1+o(1)})$
        & 2024
        & \cite{LTW}
        & Uses $\widetilde{O}(n^{1+O(\varepsilon)})$ strongly-polynomial max-flows. \\
    \hdashline
    $O(\OPT\sqrt{\log n/\varepsilon})$
        & $O(n^{O(\varepsilon)}m^{1+o(1)})$
        & \multicolumn{2}{>{\centering\arraybackslash}m{\dimexpr0.15\textwidth+2\tabcolsep\relax}}{Ours}
        & Works for any $\varepsilon\in[1/\Theta(\log n),\Theta(1)]$. \\
    \hline
    \end{tabular}
    \caption{Previous pseudo-approximation algorithms for the balanced vertex separator problem.
    $\OPT$ denotes the size of the smallest $\frac23$-balanced vertex separator.
    The notation $\widetilde{O}$ hides polylogarithmic factors.
    }
    \label{tab:previous-results}
\end{table}

\newstuff{
Approximation algorithms for balanced vertex separators have a long history.
First, it should be noted that
there is a standard reduction from the balanced vertex separator problem
to the (weighted) balanced directed edge separator problem based on vertex-splitting,
preserving approximation ratios up to constants.
For papers that study balanced directed edge separators directly, we will only discuss their implications
for vertex separators.

Polynomial-time $O(\log n)$-approximations
date back to Leighton and Rao in 1988~\cite{leighton1988approximate,leighton1999multi},
and the best-known polynomial-time approximation ratio is $O(\sqrt{\log n})$,
achieved independently by two groups at STOC 2005~\cite{FeigeHL08,agarwal2005}.
However, all of these algorithms are quite slow.
The first requires up to $O(n)$ solves of linear programs with a large number of constraints,
while the latter two require up to $O(n)$ solves of large semidefinite programs with a large number of constraints.
The number of solves is because each problem is solving the related problem of \emph{minimum-ratio vertex cut}:
Roughly speaking, this refers to finding a set of vertices $S$ whose removal disconnects a small component $C$ of the graph, while minimizing $|S|/|C|$.

Progress on fast approximation algorithms with polylogarithmic approximation ratios
has been much slower, and often involves unusual tradeoffs.
In 2007,
Arora and Kale presented an $O(\log n)$-approximation
using $\widetilde{O}(1)$ strongly-polynomial maximum flow computations~\cite{AK},
and hence runs in $\widetilde{O}(m^{1.5})$ time for a graph with $m$ edges.
It should be noted that their algorithm does very much require strongly-polynomial maximum-flow,
and does not admit an immediate replacement:
The bit-complexity of the inputs to their flow problem is polylogarithmic \emph{only} if they are represented with
floating point arithmetic.
This prevents an immediate application of almost-linear time maximum flow~\cite{linear-maxflow},
which applies for integer inputs of polynomial magnitude.
That said, we think it is feasible that a slightly-modified form of their algorithm
could be made to appropriately handle bit-complexity. This would require more work.
Arora and Kale also presented an $O(\sqrt{\log n})$-approximation running in
$\widetilde{O}(n^2+m^{1.5})$ time using multicommodity flow.

Since then, there have been a number of alternative results.
In 2010, Louis~\cite{louis2010cutmatchinggamesdirectedgraphs}
posted a preprint claiming to attain a polylogarithmic
approximation in a polylogarithmic number of maximum flow computations
using a ``cut-matching game''.
Unfortunately, the details of this proof were never published.
However,
if the suggested algorithm does work,
it does at least appear that it would not require any extra care for handling bit-complexity.

In 2019, Brandt and Wattenhofer~\cite{brandt2019approximating} presented an algorithm
with an alternative approximation guarantee:
For a graph admitting a balanced $2/3$-separator of size $\rm OPT$,
their algorithm produces a separator of size
$\widetilde{O}(OPT^2)$
in time 
$\widetilde{O}(OPT^3m)$.
This is useful for graphs with particularly small separators.

In 2024,
Lau, Tung, and Wang~\cite{LTW} published a result implying a
$O(\sqrt{\log n/\varepsilon})$-approximation for the ``vertex expansion'' problem
in $O(n^{O(\varepsilon)})$ maximum flow computations, for a parameter $\varepsilon$ similar to ours.
This consequently implies a
$O(\sqrt{\log n/\varepsilon})$-approximation for the balanced vertex separators
in $O(n^{1+O(\varepsilon)})$ maximum flow computations via a standard reduction~\cite{FeigeHL08}.
Unfortunately, their analysis omits consideration of bit-complexity,
and so faces the same issue as Arora and Kale's work
(although it is again likely possible to handle with a more careful derivation).
}

\subsection{Technical Overview}

%We describe the motivation using the vertex-separator specialization; the element-cut SDP adds edge-deletion variables treated in the same way as vertex-deletion variables.
%On the technical level, our approach is
The overall structure of our algorithm is
closest to the methods
of Arora and Kale and follow-up works~\cite{AK,Sherman,vnk:sparsest-cut}
for the balanced (edge) separator problem.
Accordingly,
we first give a brief overview of the latter problem,
and discuss the precise point at which the same approach fails,
requiring new ideas.
Our method also bears some similarities to
a slightly different approach
of Lau, Tung, and Wang~\cite{LTW}
for the directed edge expansion problem,
and we discuss where this approach fails as well.
Both of these previous works, as well as ours,
make use of the ``Matrix Multiplicative Weights Update'' algorithm (MMWU),
which is an online learning algorithm for the so-called ``Matrix game''~\cite[Algorithm 1]{AK}.
However, their implementations of this algorithm differ significantly.

\myparagraph{Balanced edge separator}
First, the methods of Arora and Kale
define an SDP relaxation of the problem that involves vector-valued variables $v_i\in\mathbb R^n$ for
each node $i\in V$.
This relaxation is then approximately solved using
a primal-dual Matrix Multiplicative Weights Update \textbf{framework}~\cite[Algorithm 2]{AK}
(called the ``Primal-Dual Algorithm for SDPs'' in the cited version of their work),
which internally runs the MMWU algorithm while maintaining the variables $v_i$ \emph{explicitly}.
The main subroutine of the primal-dual MMWU framework is a procedure called {\sc Oracle}
that does the following:  given the current primal solution (i.e.\ vectors $v_i$),
it either finds a balanced separator of small cost, or
finds an inequality violated by this solution (specified by dual variables). This inequality is then used to update vectors $v_i$.
This pair of valid outputs is key to keeping the ``width''
low, which in turn allows for a fast time complexity.

Different implementations of the oracle are presented in~\cite{AK},~\cite{Sherman} and in~\cite{vnk:sparsest-cut}.
In particular, the oracle in~\cite{vnk:sparsest-cut} works by running a sequence of maximum flow computations,
extracting directed matchings from these computations (via flow decompositions),
and then ``chaining'' together these matchings.

\myparagraph{Sparse directed edge cuts}
Lau, Tung, and Wang take an alternative approach to integrating the MMWU
algorithm into their problem.
There is a natural SDP relaxation of their problem as well,
but it would naturally involve both a vector-valued variable $v_i\in\mathbb R^n$
and a separate real-valued variable $r_i\in\mathbb R$
for each node $i\in V$,
with surprisingly challenging inequalities relating these two
(see also~\cite{lau2023cheeger}
for related SDPs with similar structure).
The standard primal-dual MMWU framework has no way to handle this structure,
so Lau, Tung, and Wang instead
use a min-max form of their problem,
with outer vector-valued variables $v_i\in\mathbb R^n$.
In their case,
the constraints on the vector-valued variables
in this min-max form happen to coincide with precisely the guarantees
that are maintained by the MMWU algorithm at all times.
They then construct another {\sc Oracle} procedure
that either solves the inner problem in a way that has low ``width'',
or outputs a combinatorial solution.

Importantly, their approach to analysis is rather ad-hoc, and only seems applicable to their specific problem,
while making constraints like balance much more difficult to handle nicely.

On the technical level, our work introduces two new ideas, which are described below.

\myparagraph{Technical contibution \#1: partial variable MMWU framework}
%We work with an SDP relaxation over variables $x_i\in\mathbb R_{\ge 0}$, $x_{ij}\in\mathbb R_{\ge 0}$ and $v_i\in\mathbb R^n$ (for $i\in V$, $ij\in E$). In the case of vertex (resp. edge) separators  variables $x_{ij}$ (resp. $x_i$) are forced to $0$. 
For simplicity, let us consider vertex separators. 
%; we  thus assume that $x_{ij}=0$.
%Our SDP relaxation in this case is a variant of the SDP for the balanced edge separator
%(and distinct from previous SDP relaxations used for the problem such as the SDP of Feige–Hajiaghayi–Lee~\cite{FeigeHL08}).
%
%\myparagraph{Balanced vertex separator}
We start by defining an SDP relaxation of the problem which is a variant of the SDP for the balanced edge separator
(and distinct from previous SDP relaxations for balanced vertex separators  such as the SDP of Feige–Hajiaghayi–Lee~\cite{FeigeHL08}).
Besides variables $v_i\in\mathbb R^n$, we use scalar variables $x_i\ge 0$
that indicate whether node $i$ belongs to the separator.

One could now try to apply the primal-dual MMWU framework described above:
maintain primal variables $v_i,x_i$, find an inequality violated by these variables,
and update these variables. Unfortunately, we were unable to find such a scheme.
Variables $v_i$ and $x_i$ play very different roles, and it is not clear how to jointly update
them while having a small ``width'' of the oracle.
We also tried to find a usable min-max formulation,
but, as mentioned, the balance requirement
prevents the same trick of
\cite{LTW}
from being applied to our problem,
since the MMWU algorithm does not automatically satisfy these constraints.

Overcoming this difficulty was the first  technical challenge that we faced.
To resolve it, we avoid the standard primal-dual MMWU framework,
and instead use a hybrid of the two approaches (which we will call the ``partial-variable'' MMWU framework):
we maintain and update variables $v_i$, but not $x_i$.
The latter are maintained neither explicitly nor implicitly, even though they are part of the primal SDP formulation.
We carefully design an oracle that does not need variables $x_i$ as input,
while still outputting dual variables,
and we apply the MMWU regret-minimization bound directly to the output of this oracle.
We also extend this framework to a more general class of SDPs.

\myparagraph{Technical contribution \#2: new oracle for the padding bound}
To prove Theorems~\ref{thm:main} and~\ref{thm:main:padding}, 
we provide two different implementations of the oracle. The oracle for the first theorem is very close to the one in~\cite{vnk:sparsest-cut};
the main difference is that we generalize
the minimum $s$-$t$ cut subroutine for edge separator to {\em element} separators.
For Theorem~\ref{thm:main:padding}
%($O(\sqrt\beta/\varepsilon)$-pseudoapproximation) 
we use a very different approach. First, we use random sampling to construct an instance of a multicommodity flow problem, which we then approximately solve using the algorithm in~\cite{ChenGraurSidford2025}. If 
the flow is large then its demand graph gives a valid output of the oracle (a feedback matrix of small width). 
The case of  small flow is the most technically challenging part of the analysis.
First, we use dual prices of the multicommodity flow problem to construct a shortest path pseudometric on $V$.
We then apply the definition of the padding bound and show the existence of Lipschitz approximation to Gaussian projections.
Using this approximation, we finally prove that an $s$-$t$ minimum cut computation for a random Gaussian direction
yields a balanced element cut of small cost with constant probability.

\subsubsection{Bit complexity}
Like previous works~\cite{AK,LTW},
we apply maximum flow as an important subroutine.
Moreover, depending on the value of the maximum flow,
we may need to produce a corresponding minimum cut.
For both vertex capacitated and directed graphs,
the fastest-known maximum flow algorithm (exact, approximate, or otherwise)
is the recent almost-linear time exact maximum flow result
of Chen et al.~\cite{linear-maxflow},
which only applies to networks with bounded integral capacities.
This presents an issue for prior work:
As presented, both the work of Arora and Kale for balanced edge separator~\cite{AK}
and its follow-up works~\cite{Sherman,vnk:sparsest-cut},
as well as the work of Lau, Tung, and Wang~\cite{LTW}
present their algorithms for real-valued weights.
Since almost-linear time maximum flow
requires bounded integers,
a bit more work would be required to complete the analysis in these works
in this paradigm.
%While the MMWU algorithm itself may be implemented
%using floating point values of small bit-complexity~\cite[Theorem 15]{musco2018stability},
%some more work would be needed to
%it is not clear whether the bit-complexity of these
%values can be bounded when converted to integers.
%In particular, the most obvious bound would be super-polynomial in magnitude
%for the work of Lau, Tung, and Wang.
%This results in polynomial overhead if the almost-linear time maximum flow result is applied,
%which is a direct consequence of the integer requirement.

In the case of balanced edge separator,
there are alternative strongly near-linear-time undirected maximum flow routines
that may be used instead~\cite{fox2023simple}.
However, in the case of directed edge expansion,
there is no such alternative.
%and we are not presently aware of any way to make the result of \cite{LTW} run
%its maximum flow computations in almost-linear time.

For our problem,
we use the element-capacitated network
to devise our maximum flow instances
in such a manner that their capacities
do not depend on the maintained variables $\{v_i\}$
in any way.
However, while this is the main step required to avoid such issues,
this does not automatically handle our bit complexity bounds,
since our capacities do still depend on a number of different quantities,
including our finite element costs,
a guess at the value of the optimum solution
(obtained via binary search),
and several constants.
To handle this cleanly,
we parametrize some of our maximum flow instance capacities
with a parameter $\mu$,
and prove several intermediate results for a value of $\mu$
falling into a wide range.
Then, our final proofs
pick an appropriate value of $\mu$
that falls into that range while having small bit-complexity.

%\input{applications.tex}

%\subsection{Related Work}
%
%\todo[inline]{We could probably briefly talk about the history of approximation algorithms for
%sparsest cut, balanced edge separators, directed sparsest cut, and directed balanced vertex separators.
%I think this section is optional, and should be relatively short if we add it.}
%\todo[inline]{Could talk about edge separators and volume-balance a bit}

\section{Preliminaries}
\label{sec:prelim}

We will now define the main convex relaxation we will work with.
Let $\xi:=\frac94c^2$,
let $\calS:=\{S\subseteq V : |S|\geq\left(1-\frac c4\right)n\}$,
and let $\calP$ be the set of sequences $p=(p_0,\ldots,p_{\ell(p)})$ of distinct nodes in $G$.
%:=\{p=(p_0,\dots,p_{l(p)}):\text{$p$ is a path in $G$}\}$.
We use the notation $p=(p_0,\dots,p_{\ell(p)})$ for the sequence of vertices along a path $p\in\calP$.
We use nonnegative scalar deletion variables $x_z$ for all elements $z\in V\cup E$; in particular, $x_{ij}$ is the variable for edge $ij$, distinct from the Gram-matrix entry $X_{ij}$.
(We write $ij$ for an undirected edge with endpoints $i,j$.)
We consider the following relaxation of the $c$-balanced element cut problem:

\begin{subequations}\label[problem]{eq:SDP}
\begin{align}
\min & \sum_{z\in V\cup E}w(z)x_z   &   \label{eq:SDP:a} \\
x_i+x_j+x_{ij} & \ge \|v_i-v_j\|^2                                                                               &x_i+x_j+x_{ij}-L_{ij}\bullet X&\ge 0                           && \forall ij\in E \label{eq:SDP:b} \\
\|v_i\|^2&=1                                                                               &X_{ii}&=1                           && \forall i\in V \label{eq:SDP:c} \\
\sum_{j=1}^{\ell(p)} \|v_{p_j}-v_{p_{j-1}}\|^2 & \ge \| v_{p_{\ell(p)}}-v_{p_0}\|^2           &T_p\bullet X &\ge 0                   && \forall p\in\calP \label{eq:SDP:d} \\
\sum_{i,j\in S:i<j}\|v_i-v_j\|^2 &\ge \xi n^2                                                   &K_{S}\bullet X&\ge \xi n^2                 && \forall S\in\calS \label{eq:SDP:e} \\
            x               &\ge 0                                                             &X\succeq  0,x&\ge 0 \label{eq:SDP:f}
\end{align}
\end{subequations}

Variables for elements of infinite cost are fixed to zero and their objective terms are omitted.
In particular, at the optimum we have $x_i=0$, $x_{ij}=\|v_i-v_j\|^2$ in the case of edge separators and $x_{ij}=0$ in the case of vertex separators.
The left and right sides here both denote the same optimization problem.
The difference is notational.
The matrix $X$ is defined as $X:=V^T V$, where $V$ is the $n\times n$ matrix with columns $v_1,\ldots,v_n$.
We have  $L_{ij}=\calL(G_{(i,j)})$
and $T_p=\calL(G_p)-\calL(G_{(p_0,p_{\ell(p)})})$
where $\calL(\cdot)$ is the Laplacian of the corresponding graph (the diagonal degree matrix minus the adjacency matrix) and $G_q$ for a path $q$ is the undirected unweighted graph containing all edges of $q$.
Finally, we define $K_{S}(i,i)=|S|-1$ for $i\in S$,
$K_{S}(i,j)=-1$ for distinct $i,j\in S$,
and $K_{S}(i,j)=0$ in all other cases.
%{\color{red} [The diagonal of $K_S$ is zero, so it's not a Laplacian] That is, $K_{S}$ is the Laplacian of the complete
%graph over $S$.} % Gotcha

\begin{proposition}
The optimal value of~\cref{eq:SDP} divided by 4 is a lower bound on the minimum $c$-{\sc Balanced Element Cut} problem.
\end{proposition}

\begin{proof}
Let $(A,B),C$ be an optimal $c$-balanced element cut,
so $\max\{|A|,|B|\}\leq(1-c)n$. We need to construct a feasible solution of \cref{eq:SDP} of value at most $4(w(C)+w(E(A,B)))$.
In the construction below, the last $n-1$ coordinates of vectors $v_i$ will be set to $0$; we can thus assume for brevity of notation that $v_i\in\mathbb R$ rather than $v_i\in\mathbb R^n$.

Let us choose a partition $(\hat A,\hat B)$ of $V$ with $A\subseteq \hat A$, $B\subseteq \hat B$ which is the ``most balanced'', i.e.\  
it minimizes the absolute value of $|\hat A|-|\hat B|$.
In particular, if $\max\{|A|,|B|\}\le n/2$
then $\min\{|\hat A|,|\hat B|\}\ge\lfloor n/2 \rfloor$,
otherwise either 
$(\hat A,\hat B)=(A,V- A)$ (if $|A|>n/2$) or 
$(\hat A,\hat B)=(V- B,B)$ (if $|B|>n/2$).
In all three cases we have $|\hat A|\ge c n$ and $|\hat B|\ge cn$, since $c<1/2$.
Set 
$x_i=\begin{cases}0 & \mbox{if }i\in A\cup B \\ 4&\mbox{if }i\in C\end{cases}$ 
and 
$v_i=\begin{cases}-1 & \mbox{if }i\in \hat A \\ +1&\mbox{if }i\in \hat B\end{cases}$.
Set $x_{ij}=4$ for $ij\in E(A,B)$ and $x_{ij}=0$ otherwise.
Checking conditions~\eqref{eq:SDP:b}, \eqref{eq:SDP:c}, \eqref{eq:SDP:d}, \eqref{eq:SDP:f} is straightforward.
(For each edge $ij$ with $v_i\ne v_j$, either an endpoint lies in $C$, or $ij\in E(A,B)$, so $x_i+x_j+x_{ij}\ge4$.) Let us show~\eqref{eq:SDP:e}. Consider a set $S\subseteq V$
with $|S|\ge (1-\frac c 4) n$.
We have
        $|\hat A\cap S|
        %=|A|-|A- S|
        %\geq|A|-|V- S|
        \geq|\hat A|-\frac c4n\geq\frac{3c}4n$,
        and similarly for $\hat B\cap S$.
        Therefore,
        \begin{equation*}
        \begin{aligned}
        \sum_{i,j\in S:i<j}\|v_i-v_j\|^2
         \geq & \sum_{i\in \hat A\cap S,j\in \hat B\cap S}\|v_i-v_j\|^2
         =  \sum_{i\in \hat A\cap S}\sum_{j\in \hat B\cap S}4
         =  4|\hat A\cap S|\,|\hat B\cap S|
         \geq  4\left(\frac{3c}4n\right)^2
         %=  \frac94c^2n^2
         %=  \xi n^2,
        \end{aligned}
        \end{equation*}
        The last expression equals $\frac94c^2n^2
         =  \xi n^2$, and so eq. \eqref{eq:SDP:e} is indeed satisfied.
\end{proof}

The dual of~\cref{eq:SDP} is as follows. It has variables $y_i$ for every node $i\in V$, $\lambda_{ij}$ for every edge $ij\in E$,
$f_p$ for every path $p$, and $z_S$ for every set $S$ of size at least $(1-c/4)n$.
Let ${\tt diag}(y)$ be the diagonal matrix with vector $y$ on the diagonal:
\begin{subequations}\label[problem]{eq:SDPdual}
\begin{align}
\max \sum_i y_i + \xi n^2\sum_S z_S \label{eq:SDPdual:a} \\
{\tt diag}(y)  + \sum_p f_p T_p + \sum_S z_S K_{S} &\preceq \sum_{ij} \lambda_{ij}L_{ij} \label{eq:SDPdual:b} \\
\sum_{j:ij\in E} \lambda_{ij} &\le \midweight(i) \qquad\quad\forall i  \label{eq:SDPdual:c} \\
\lambda_{ij} &\le w(ij) \qquad\quad\forall ij\in E \label{eq:SDPdual:edge} \\
f_p,z_S,\lambda_{ij} &\ge 0 \qquad\quad \forall p,S,ij \label{eq:SDPdual:d}
\end{align}
\end{subequations}
For a vector $\lambda\ge 0$ let $G^\lambda\eqdef (V,E,\lambda)$ be the graph with edge weights $\lambda$.
Note that $\lambda$ satisfies~\cref{eq:SDPdual:c} if and only if $\deg(G^\lambda)\le w|_V$ component-wise,
where $\deg(G^\lambda)\in\mathbb R^V$ is the vector of node degrees in $G^\lambda$.
Also, $\sum_{ij} \lambda_{ij}L_{ij}=\calL(G^\lambda)$ is the Laplacian of $G^\lambda$.

\section{A new matrix multiplicative weights method }\label{sec:MW}
%To solve the \cref{eq:SDP},
%we will somewhat follow the methods of Arora and Kale
%for solving the SDP relaxation of the
%$c$-{\tt Balanced Separator} problem~\cite{AK}.
%We will also apply related constructions and improvements for the sparsest cut problem
%used by Sherman, Kolmogorov, and Lau, Tung and Wang~\cite{Sherman,LTW,vnk:sparsest-cut}.

As previously discussed, we are unable to apply the general primal-dual framework of Arora and Kale~\cite{AK}.
However, we still leverage some of the same high-level steps.
Instead of solving optimization~\cref{eq:SDP} directly, we will use a binary search on the objective value.
Accordingly, let us fix value $\alpha\in[1,8W]$. We will present 
a procedure that will either (i) certify that the optimum value of~\cref{eq:SDP} is at least $\alpha/2$, or (ii) produce a $\Theta(1)$-balanced element cut of cost at most $\Theta(\kappa\alpha)$, for some parameter $\kappa$ to be specified later.
Clearly, by running this procedure
for $O(\log n)$ values of $\alpha$ we will obtain a $O(\kappa)$ pseudo-approximation to the minimum element cut problem.

The procedure above will be implemented via the Matrix Multiplicative Weights Update (MMWU) algorithm
of Arora and Kale~\cite[Algorithm 1]{AK} given below.

\begin{algorithm}[H]
  \DontPrintSemicolon
\SetNoFillComment
\For{$t=1,2,\ldots,T$}
{
	Compute 
$$
W^{(t)}=\exp\left(\eta\sum_{r=1}^{t-1}N^{(r)}\right)\;,\quad
X^{(t)} = 
n \cdot \frac{W^{(t)}} {{\tt Tr}(W^{(t)})}
$$\\
Either output {\tt Fail}, or find ``feedback matrix'' $N^{(t)}$ of the form
\begin{eqnarray*}
 N^{(t)}&=&{\tt diag}(y)+\sum_p f_pT_p + \sum_S z_S K_{S} - \sum_{ij}\lambda_{ij}L_{ij} 
 \end{eqnarray*}
%with $\|M^{(t)}\|\le 1$ such that $M_t \bullet X^{(t)}>0$ and $M_t\bullet X\le 0$ for all feasible $X$
%(or output ``fail''). Matrix $M^{(t)}$ should have the form
 where $f_p\ge 0$, $z_S\ge 0$, $\sum_i y_i+\xi n^2\sum_Sz_S\ge \alpha$, $\lambda_{ij}$ are non-negative variables with
% $\sum_{j:ij\in E}\lambda_{ij}\le 1$ for each $i$,
 $\deg(G^\lambda)\le w|_V$ and $\lambda_{ij}\le w(ij)$,
 and $N^{(t)}\bullet X^{(t)}\le 0$.
}

      \caption{MMWU algorithm. 
      }\label{alg:MW}
\end{algorithm}

%Note that we do not explicitly maintain variables $x_i$ in \cref{eq:SDP}.

\begin{theorem}
\label{thm:regret-bound}
Set $\eta=\frac{\delta}{2n\rho^2}$ and $T=\lceil\frac{4n^2\rho^2\ln n}{\delta^2}\rceil$.
Suppose that the algorithm does not fail during the first $T$ iterations, and the spectral norms are bounded as $\|N^{(t)}\|\le \rho$
for all $t\in[T]$.
Then optimal values of \cref{eq:SDP,eq:SDPdual}
are at least $\alpha-\delta$.
\end{theorem}
%\begin{proof}
%Let us apply the standard MMWU regret bound \cite[Corollary 3.2]{AK} with matrices $P^{(t)}=\frac{1}{n} X^{(t)}$ and $M^{(t)}=-\frac 1\rho N^{(t)}$.
%We obtain
%\[
%\sum_{t=1}^T M^{(t)}\bullet P^{(t)} \le \lambda_{n}\left(\sum_{t=1}^TM^{(t)}\right) + \eta\rho \cdot T + \frac {\ln n} {\eta \rho}
%\]
%where $\lambda_1(A)\ge\cdots\ge\lambda_{n}(A)$ denote eigenvalues of symmetric matrix $A$. Equivalently,
%\[
%\frac 1 \rho \lambda_{n}\left(-\sum_{t=1}^T N^{(t)}\right) 
%\ge -\frac 1 {n\rho} \sum_{t=1}^T N^{(t)}\bullet X^{(t)} - \eta \rho T  - \frac {\ln n} {\eta \rho}
%\]
%Define variables $y_i, f_p, z_S, \lambda_{ij}$ via
%$y_i=-\frac \delta n +\frac 1T\sum_t y^{(t)}_i$, 
%$f_p=\frac 1T\sum_t f^{(t)}_p$, 
%$z_S=\frac 1T\sum_t z^{(t)}_S$, 
%$\lambda_{ij}=\frac 1T\sum_t \lambda^{(t)}_{ij}$.
%Let $N={\tt diag}(y)+\sum_p f_pT_p + \sum_S z_S K_{S} - \sum_{ij}\lambda_{ij}L_{ij}$.
%We can write
%\[
%\lambda_n(-N)=\lambda_n\left(\frac \delta n I-\frac 1T \sum_{t=1}^TN^{(t)}\right)
%=\frac \delta n + \frac 1T \lambda_n\left(- \sum_{t=1}^TN^{(t)}\right)
%\ge \frac \delta n - \eta \rho^2 - \frac{\ln n}{\eta T} 
%\]
%Using definitions of $\eta$ and $T$, we conclude that the last expression is nonnegative, and hence $N \preceq 0$.
%This means that $y_i, f_p, z_S, \lambda_{ij}$
%is a feasible solution of the dual \cref{eq:SDPdual}; both vertex and edge capacity constraints are preserved by averaging.
%Therefore, optimal values of \cref{eq:SDP,eq:SDPdual}
%are at least $\sum_i y_i + \xi n^2 \sum_S z_S\ge -n \cdot \frac \delta n  + \alpha=\alpha-\delta$,
%as claimed.
%\end{proof}
We will set $\delta=\frac 12 \alpha$, and 
require 
 Algorithm~\ref{alg:MW} to 
 produce a $\Theta(1)$-balanced element cut of cost at most $\Theta(\kappa \alpha)$ whenever it outputs  {\tt Fail}.

\subsection{Partial-variable MMWU framework}
To prove Theorem~\ref{thm:regret-bound}, 
we will first derive an MMWU framework for more general SDPs, and then specialize it to Algorithm~\ref{alg:MW}.
%Recall that in Algorithm~\ref{alg:MW} we do not maintain the scalar variables $x_z$, $z\in V\cup E$. In order to explain why it works, in this section we describe such a scheme
%for more general SDPs. 

Let us consider an SDP with a block-diagonal symmetric-matrix $
X = \operatorname{blkdiag}(X_1,\dots,X_K)$, $X_k \in \mathbb{S}^{n_k}
$
given by 
%The primal semidefinite program is
\begin{subequations}\label{eq:PP}
\begin{align}
\min_{X_1,\dots,X_K} \quad
& \sum_{k=1}^K \langle C_k, X_k \rangle \\
\text{s.t.} \quad
& \sum_{k=1}^K \langle A_{ik}, X_k \rangle \ge b_i  && \hspace{-40pt} \forall i\in\calI   \label{eq:PP:b} \\
& \sum_{k=1}^K \langle A_{ik}, X_k \rangle = b_i    && \hspace{-40pt} \forall i\in\calI'  \label{eq:PP:c} \\
& X_k \succeq 0                                     && \hspace{-40pt} \forall k\in[K]
\end{align}
\end{subequations}

The corresponding dual problem is
\begin{subequations}\label{eq:DD}
\begin{align}
\max_{y} \quad
& b^\top y \\
\text{s.t.} \quad
& S_k(y) \eqdef C_k - \sum_i y_i A_{ik} \succeq 0 \quad & \hspace{-40pt} \forall k\in[K] \\
& y \in %\calY \eqdef 
\mathbb R_{\ge 0}^\calI \times \mathbb R^{\calI'} &
\end{align}
\end{subequations}
where $\sum_i$ means the sum over $i\in\calI\cup\calI'$.

%Equivalently, the dual slack matrix is block-diagonal:
%\begin{equation}
%S(y) = \operatorname{blkdiag}(S_1(y),\dots,S_K(y)) \succeq 0.
%\end{equation}

%Let us fix a subset $\calK\subseteq [K]$. 

We will describe an algorithm that
maintains variables $X_k$ only for some blocks $k$.
The set of such blocks will be denoted as $\calK\subseteq [K]$.

We assume that there exists $R>0$ such that for each $k\in\calK$ the linear constraints \eqref{eq:PP:b}, \eqref{eq:PP:c} of the primal problem
imply that ${\tt Tr}(X_k)\le R$, i.e.\ a linear combination of these constraints yields constraint ${\tt Tr}(X_k)\le R$. We further assume that if the $i$-constraint is present in this linear combination with a non-zero coefficient then $A_{i\ell}=0$ for all $\ell \ne k$.

\begin{algorithm}[H]
  \DontPrintSemicolon
\SetNoFillComment
\For{$t=1,2,\ldots,T$}
{
	For each $k\in\calK$ compute 
$$
W^{(t)}_k=\exp\left(\eta\sum_{r=1}^{t-1}N^{(r)}_k\right)\;,\quad
X^{(t)}_k = 
R \cdot \frac{W^{(t)}_k} {{\tt Tr}(W^{(t)}_k)}
$$\\
Either output {\tt Fail}, or find vector $y=y^{(t)}\in \mathbb R_{\ge 0}^\calI \times \mathbb R^{\calI'}$ such that $b^\top y\ge \alpha$ and
\begin{subequations}
\begin{align}
N_k^{(t)}\bullet X_k^{(t)}&\le 0 ,\qquad N_k^{(t)}\eqdef -S_k(y) && \forall k\in\calK \\
S_k(y)& \succeq 0 && \forall k\in[K]-\calK
\end{align}
\end{subequations}
}

      \caption{Partial-variable MMWU algorithm. 
      }\label{alg:MW:general}
\end{algorithm}

\begin{theorem}
\label{thm:regret-bound:general}
Set $\eta=\frac{\delta}{2R\rho^2}$ and $T=\lceil\frac{4R^2\rho^2\ln n_{\max}}{\delta^2}\rceil$ where $n_{\max}=\max_{k\in\calK} n_k$.
Suppose that the algorithm does not fail during the first $T$ iterations, and the spectral norms are bounded as $\|N^{(t)}_k\|\le \rho$
for all $k\in\calK$ and $t\in[T]$.
Then optimal values of problems \eqref{eq:PP}, \eqref{eq:DD}
are at least $\alpha-\delta\cdot|\calK|$.
\end{theorem}
Note that the standard MMWU framework is obtained by setting
$\calK=[K]=\{1\}$. In this case Theorem~\ref{thm:regret-bound:general} reduces to (a special case of \footnote{
\cite[Theorem 4.6]{AK} allows the oracle to set $N^{(t)}_k=\sum_iy_i A_{ik}-F^{(t)}_k$ for some matrix $F^{(t)}\preceq C_k$,
instead of $N^{(t)}_k=\sum_iy_i A_{ik}-C_k$.
The same could be done in our case, but since our algorithm does not need this extra flexibility, for simplicity we state
the theorem without this generalization. })
 Theorem 4.6 in~\cite{AK}.

\begin{proof}
Since the constraints
\eqref{eq:PP:b}, \eqref{eq:PP:c}
imply constraint $-{\tt Tr}(X_k)\ge - R$,
there exist coefficients $\gamma_{ik}$ with $\gamma_{ik}\ge 0$ for $i\in\calI$ such that for each $k\in\calK$ we have
$\sum_i\gamma_{ik} A_{ik}=-I$ and $\sum_i \gamma_{ik}b_i\ge -R$.
Define $\bar y=\frac 1 T \sum_{t=1}^T y^{(t)}$,
and let $y$ be the vector
with $y_i=\bar y_i+\sum_{k\in\calK} \gamma_{ik}\frac{\delta}{R}$ for each $i$.
Clearly, we have $y\in \mathbb R_{\ge 0}^\calI \times \mathbb R^{\calI'}$.

Note that $b^\top\bar y=\frac 1T\sum_{t=1}^T b^\top y^{(t)} \ge \alpha$.
We can also write
$$
b^\top y=b^\top \bar y + \sum_{k\in\calK}\sum_i b_i\gamma_{ik}\frac{\delta}{R}\ge b^\top\bar y-\sum_{k\in\calK} R\frac{\delta}{R}
\ge \alpha-\delta\cdot|\calK|
$$
To prove the claim, it thus now suffices to show that vector $y$ is a feasible solution of the dual problem~\eqref{eq:DD},
i.e.\ that $S_k(y)\succeq 0$ for each $k\in [K]$.

Using the assumption before Algorithm~\ref{alg:MW:general}, we obtain
\begin{align*}
S_k(y)&=S_k(\bar y)-\sum_i\gamma_{ik}\frac \delta{R} A_{ik}
%S_k(\bar y)-\frac{\delta}{R}(-I)
=S_k(\bar y)+\frac \delta R I
&& \forall k\in\calK \\
S_k(y)&=S_k(\bar y) && \forall k\in [K]-\calK
%=-\frac 1 T \sum_{t=1}^T N^{(t)}_k +\frac \delta R I
\end{align*}

First, let us consider $k\in\calK$.
Let us apply the standard MMWU regret bound \cite[Corollary 3.2]{AK} with matrices $P^{(t)}=\frac{1}{R} X^{(t)}_k$ and $M^{(t)}=-\frac 1\rho N^{(t)}_k$.
We obtain
\[
\sum_{t=1}^T M^{(t)}\bullet P^{(t)} \le \lambda_{n}\left(\sum_{t=1}^TM^{(t)}\right) + \eta\rho \cdot T + \frac {\ln n} {\eta \rho}
\]
where $n=n_k$ and $\lambda_1(A)\ge\cdots\ge\lambda_{n}(A)$ denote eigenvalues of symmetric matrix $A$. Equivalently,
\[
\frac 1 \rho \lambda_{n}\left(-\sum_{t=1}^T N^{(t)}_k\right) 
\ge \underbrace{-\frac 1 {R\rho} \sum_{t=1}^T N^{(t)}_k\bullet X^{(t)}_k}_{\ge 0} - \eta \rho T  - \frac {\ln n} {\eta \rho}
\]
By noting that $- \sum_{t=1}^T N^{(t)}_k=T S_k(\bar y)$ we obtain
$$
\lambda_n(S_k(y))=\frac 1 T \lambda_n\left( T S_k(\bar y)\right)+\frac \delta R
\ge \frac 1 T\left(-\eta\rho^2 T - \frac{\ln n}\eta\right) +\frac \delta R
= \frac \delta R - \eta \rho^2 - \frac{\ln n}{\eta T} 
$$
Using definitions of $\eta$ and $T$, we conclude that the last expression is nonnegative, and hence $S_k(y) \succeq 0$.

Now consider $k\in[K]-\calK$. We have $S_k(y^{(t)})\succeq 0$ for each $t$ and hence $S_k(y)=S_k(\bar y)\succeq 0$.

\end{proof}

\begin{proof}[Proof of Theorem~\ref{thm:regret-bound}]
Apply Theorem~\ref{thm:regret-bound:general} to~\cref{eq:SDP} by treating its Gram matrix $X$ as one $n\times n$ block and every finite-cost scalar variable $x_z$ as a $1\times1$ block.  Maintain only the Gram-matrix block, so $\calK$ consists of this single block.  The constraints $X_{ii}=1$ imply ${\tt Tr}(X)=n$, and hence we take $R=n$ and $n_{\max}=n$.  The update for the maintained block in Algorithm~\ref{alg:MW:general} is then exactly the update in Algorithm~\ref{alg:MW}.

For the Gram-matrix block, the matrix $N^{(t)}=-S_X(y^{(t)})$ is precisely the feedback matrix
$
 {\tt diag}(y)+\sum_p f_pT_p+\sum_S z_SK_S-\sum_{ij}\lambda_{ij}L_{ij}
$
used in Algorithm~\ref{alg:MW}.  Positive semidefiniteness of the unmaintained $1\times1$ slacks says exactly that
$\deg(G^\lambda)\le w|_V$ and $\lambda_{ij}\le w(ij)$, while $b^\top y\ge\alpha$ is the displayed lower bound on the dual objective.  Thus the oracle conditions in the two algorithms agree.  Finally, $|\calK|=1$, so Theorem~\ref{thm:regret-bound:general} gives the lower bound $\alpha-\delta$; substituting $R=n$ and $n_{\max}=n$ also gives
$
 \eta=\frac{\delta}{2n\rho^2}
$,
$
 T=\left\lceil\frac{4n^2\rho^2\ln n}{\delta^2}\right\rceil
$.
These are exactly the conclusion and parameter choices of Theorem~\ref{thm:regret-bound}.
\end{proof}

\section{Implementation of Algorithm~\ref{alg:MW}}
\label{sec:low-rank}
%A naive implementation of the MMWU algorithm would not be efficient.
%First, we need to work with matrices of size $n\times n$, so the runtime and memory usage would be $\Omega(n^2)$.
%Furthermore, the algorithm requires computing matrix exponents which can be a challenging task.
%Arora and Kale showed in~\cite{AK} how to address both problems.
%Instead of worki
%In this section we discuss how to implement one step of the algorithm above.

Consider the matrix $X=X^{(t)}$ computed at the $t$-th step of \cref{alg:MW}.
It can be seen that $X$ is positive semidefinite, so we can consider its Gram decomposition: $X=V^TV$.
Let $v_1,\ldots,v_n$ be the columns of $V$; clearly, they uniquely define $X$. 
Note that computing $v_1,\ldots,v_n$ requires matrix exponentiation, which is a tricky operation because of accuracy issues.
Furthermore, even storing these vectors requires $\Theta(n^2)$ space and thus $\Omega(n^2)$ time, which is too slow for our purposes.
To address these issues, Arora and Kale compute approximations $\tilde v_1,\ldots,\tilde v_n$ to these vectors
using the following result.
\begin{theorem}[{\cite[Lemma 7.2]{AK}}]\label{th:GramApproximation}
For any constant $c>0$ there exists an algorithm that does the following:
given values $\gamma\in(0,\tfrac 12)$, $\lambda_{\max}>0$, $\tau=O(n^{3/2})$ and
(implicit) symmetric matrix $A\in\mathbb R^{n\times n}$ of spectral norm $\|A\|\le \lambda_{\max}$, it computes a matrix $\tilde V\in\mathbb R^{d\times n}$ with column vectors
$\tilde v_1,\ldots,\tilde v_n$ 
of dimension $d=O(\tfrac{\log n}{\gamma^2})$ such that matrix $\tilde X=\tilde V^T\tilde V$ has trace $n$,
and with probability at least $1-n^{-c}$,
one has
\begin{subequations}\label{eq:GramApproximation}
\begin{align}
|\;\|\tilde v_i\|^2-\|v_i\|^2\;| &\;\;\le\;\; \gamma (\|\tilde v_i\|^2+\tau) && \forall i \\
|\;\|\tilde v_i-\tilde v_j\|^2-\|v_i-v_j\|^2\;| &\;\;\le\;\; \gamma (\|\tilde v_i-\tilde v_j\|^2+\tau) && \forall i,j
\end{align}
\end{subequations}
where $v_1,\ldots,v_n$ are the columns of a Gram decomposition of $X=n\cdot \frac{\exp(A)}{{\tt Tr}(\exp(A))}$.
The complexity of this algorithm equals the complexity of computing $kd$ matrix-vector products of the form $A\cdot u$, $u\in\mathbb R^n$,
where $k=O(\max\{\lambda_{\max}^2,\log \frac{n^{5/2}}{\tau}\})$.
\end{theorem}
If we use this theorem inside Algorithm~\ref{alg:MW}, then the matrices
$A$ will have the form $A=\eta\sum_{r=1}^{t-1}N^{(r)}$;
their spectral norm will be bounded by $\eta\rho T$.
Therefore, we can set $\lambda_{\max}=\eta\rho T=\Theta(\frac{\rho n \log n}{\alpha})$
in \cref{th:GramApproximation}.
Parameters $\gamma$ and $\tau$ will be specified later.

We will write $(v_1,\ldots,v_n)\approx_{\gamma,\tau}(\tilde v_1,\ldots,\tilde v_n)$ if conditions~\eqref{eq:GramApproximation} hold.
We now need to show how to solve the following problem.

\vspace{5pt}
\noindent\hspace{0pt}
\begin{minipage}{\dimexpr\columnwidth-10pt\relax}
\fbox{\parbox{\textwidth}{

\noindent {\bf Input}:
 (unobserved) matrix $V\in\mathbb R^{n\times n}$ with columns $v_1,\ldots,v_n\in\mathbb R^n$ and 
 (observed) matrix $\tilde V\in\mathbb R^{d\times n}$ with columns $\tilde v_1,\ldots,\tilde v_n\in\mathbb R^d$ such
 that 
  $(v_1,\ldots,v_n)\approx_{\gamma,\tau}(\tilde v_1,\ldots,\tilde v_n)$ and
${\tt Tr}(X)={\tt Tr}(\tilde X)=n$ where 
$X=V^TV$ and $\tilde X=\tilde V^T\tilde V$.

\noindent {\bf Output}: either 
\\ (i) matrix $N$ of the form 
$ N={\tt diag}(y)+\sum_p f_pT_p + \sum_S z_S K_{S} - \sum_{ij}\lambda_{ij}L_{ij} 
$
 where $f_p\ge 0$, $z_S\ge 0$, $\sum_i y_i+\xi n^2\sum_Sz_S\ge \alpha$, $\lambda_{ij}$ are non-negative variables with
% $\sum_{j:ij\in E}\lambda_{ij}\le 1$ for each $i$,
 $\deg(G^\lambda)\le w|_V$ and $\lambda_{ij}\le w(ij)$,
  and $N\bullet X\le 0$; or
 \\
(ii) a $\Theta(1)$-balanced element cut of value at most $\kappa\alpha$.

}}
\end{minipage}
\vspace{5pt}

A procedure that solves the problem above will be called an ``{\tt Oracle}'',
and the spectral norm $\|N\|$ of matrix $N$ will be called the {\em width} of the oracle.

%\newpage % page break here so that the spacing around the long Matching algorithm is a little more consistent under edits
%\section{{\tt Oracle} implementation}
%\label{sec:oracle-impl}
%
%In this section, we will show how to implement the
%oracle used for MMWU.
%That is, we will solve the problem specified at the end of the previous section.

%We will implement the oracle first by handling an easy case that returns early if triggered.
First, we will handle an easy case that returns early if triggered.

Let us denote $S=\{i\in V\::\:\|\tilde v_i\|^2\le 4/c\}$. 
We have $\sum_{i\in V}\|\tilde v_i\|^2={\tt Tr}(\tilde V^T\tilde V)=n$ and thus $|S|\ge (1-c/4)n$.

\begin{proposition}\label{prop:oracle-easy-case}
Suppose that $K_S\bullet \tilde X< \tfrac {\xi n^2}4$.
Then setting $y_i=-\tfrac{\alpha}n$ for all $i\in V$, $z_{S}=\tfrac{2\alpha}{\xi n^2}$, $z_{S'}=0$ for all $S'\ne S$, and $\lambda_{ij}=0$
gives a valid output of the oracle with width $\rho=O(\tfrac \alpha n)$ assuming that 
 parameters $\tau,\gamma$ in \cref{th:GramApproximation} satisfy 
 $\gamma\le \tfrac 12$ and $\tau\le \tfrac \xi 2$.
\end{proposition}
\begin{proof}
Denote $z_{ij}=\|v_i-v_j\|^2$ and $\tilde z_{ij}=\|\tilde v_i-\tilde v_j\|^2$.
We know that $\tilde Z:=\sum_{ij} \tilde z_{ij}=K_{S}\bullet \tilde X<\tfrac {\xi n^2}4$ where the sum is over $i,j\in S$ with $i<j$.
Also, $|z_{ij}-\tilde z_{ij}|\le \gamma(\tilde z_{ij}+\tau)$ for all $i,j$.
This implies that 
$$
K_{S}\bullet X=\sum_{i,j\in S:i<j} z_{ij}
< \tilde Z + \gamma \tilde Z + \tfrac{n^2}{2}\gamma\tau
\le(1+\gamma) \tfrac {\xi n^2}4 + \tfrac{n^2}{2} \gamma\tau 
\le \tfrac{\xi n^2}2
$$
Note that $N=-\tfrac \alpha n I + \tfrac {2\alpha}{\xi n^2}K_{S}$.
We have $\sum_iy_i+\xi n^2\sum_{S'}z_{S'}=n\cdot(-\tfrac \alpha n)+\xi n^2\cdot \tfrac {2\alpha}{\xi n^2}=\alpha$
and $N\bullet X=-\tfrac \alpha n I\bullet X + \tfrac {2\alpha}{\xi n^2}(K_{S}\bullet X)
\le -\tfrac \alpha n \cdot n + \tfrac {2\alpha}{\xi n^2}\cdot \tfrac {\xi n^2}2=0$, as desired.
Also, $\|N\|\le\|-\tfrac \alpha n I\|+\|\tfrac {2\alpha}{\xi n^2}K_{S}\|=O(\tfrac \alpha n)$.
\end{proof}

From now on we make the following assumption.
\begin{assumption}\label{assumption:main}
$\|\tilde v_i\|^2\le 4/c$ for all $i\in S$, $|S|\ge (1-c/4)n$ and $\sum\limits_{i,j\in S:i<j}\|\tilde v_i-\tilde v_j\|^2=K_S\bullet \tilde X\ge \tfrac{\xi n^2}4$.
\end{assumption}

Let us set $y_i=\frac \alpha n$ for all $i\in V$ and $z_{S'}=0$ for all $S'$, then $\sum_i y_i+\xi n^2\sum_{S'} z_{S'}=\alpha$.
Note that $N=\tfrac \alpha n I + \sum_p f_p T_p - \sum_{ij}\lambda_{ij}L_{ij}$ and $\tfrac \alpha n I\bullet X=\tfrac \alpha n {\tt Tr}(X)=\alpha$,
so condition $N\bullet X\le 0$ is equivalent to
$(\sum_p f_p T_p  - \sum_{ij}\lambda_{ij}L_{ij})\bullet X\le -\alpha$.
With a variable substitution, our goal thus becomes as follows.\footnote{
An analogous specification (for the balanced edge separator problem) appeared in~\cite{vnk:sparsest-cut},
but in~\cite{vnk:sparsest-cut} it was further assumed that $S=V$.
However, the justification for this simplification in Appendix A.2 is incorrect:
it works if the oracle returns a feedback matrix, but does not work if the oracle returns a cut.
A similar simplification claim also appeared in~\cite{Sherman}. In this work we do not assume that $S=V$.
}

\vspace{5pt}
\noindent\hspace{0pt}
\begin{minipage}{\dimexpr\columnwidth-10pt\relax}
\fbox{\parbox{\textwidth}{
 \em Find either\\
 (i)  
a matrix $N$ of the form 
$ N=\sum_p f_pT_p  - \sum_{ij}\lambda_{ij}L_{ij} 
$
 where $f_p\ge 0$, $\lambda_{ij}\ge 0$,
% $\sum_{j:ij\in E}\lambda_{ij}\le 1$ for each $i$,
 $\deg(G^\lambda)\le w|_V$ and $\lambda_{ij}\le w(ij)$,
  and $N\bullet X\le -\alpha$;  
  or\\
  (ii) a $\Theta(1)$-balanced element cut of value at most $\kappa\alpha$.
}}
\end{minipage}
\vspace{5pt}

A procedure for solving this problem will be called a ``reduced {\sc Oracle}'', and the norm $\|N\|$ will be its width.
Note that the width of the full oracle is at most $\tfrac \alpha n$ plus the width of the reduced oracle, since $\|\tfrac \alpha n I + N\| \le \tfrac \alpha n + \|N\|$.
We will show the following.

\begin{theorem}\label{thm:reduced-oracle}
Assume that the reduced oracle has the same input as {\sc Oracle} that additionally satisfies Assumption~\ref{assumption:main}.
Then it can be implemented so that it
\begin{itemize}
\item runs in $m^{1+o(1)}n^{O(\varepsilon)}\operatorname{polylog}(n)$ time and has parameter $\kappa=\Theta(\sqrt{\beta}/\varepsilon)$, assuming that parameters 
$\Gamma$ and $\beta\in[1,O(\sqrt{\log n})]$  of the padded bound are known;
\item runs in $O\left(n^{O(\varepsilon)}\cdot\tflow(m,m)\cdot{\rm polylog}(n)\right)$ time and has parameter $\kappa=\Theta(\sqrt{\log n/\varepsilon})$.
\end{itemize}
In both cases the oracle succeeds (i.e. produces either output (i) or (ii)) with probability at least $1-n^{-\Theta(1)}$.
Both cases use $\tau=\Theta(1)$ and $\gamma=\Omega(1/\log n)$.  Matrix $N$ satisfies
$\|N\|\le \alpha n^{-1+O(\varepsilon)}\operatorname{polylog}(n)$ and has
$O(mn^{O(\varepsilon)}\operatorname{polylog}(n))$ non-zero entries.
\end{theorem}

Assuming this result, we can easily prove Theorems~\ref{thm:main} and \ref{thm:main:padding}.
\begin{proof}[Proof of \cref{thm:main,thm:main:padding}]
Theorem~\ref{thm:reduced-oracle} and Proposition~\ref{prop:oracle-easy-case} 
give oracle implementations with and without Assumption~\ref{assumption:main}, respectively.
In both cases the oracle width is $\rho=O(\alpha n^{-1+O(\varepsilon)}\operatorname{polylog}(n))$.
Thus, the number of iterations in Theorem~\ref{thm:regret-bound} with $\delta=\alpha/2$ is
$ T=O\left(\frac{n^2\rho^2\log n}{\alpha^2}\right)
   =n^{O(\varepsilon)}\operatorname{polylog}(n)$.
Moreover,
$\eta\rho T=O(\rho n\log n/\alpha)
=n^{O(\varepsilon)}\operatorname{polylog}(n)$.  Since
$\gamma=\Omega(1/\log n)$, Theorem~\ref{th:GramApproximation} uses dimension
$d=O(\log^3 n)$ and
$n^{O(\varepsilon)}\operatorname{polylog}(n)$ matrix-vector products per
iteration. Each such product can be computed in
$O(mn^{O(\varepsilon)}\operatorname{polylog}(n))$ times, so we need 
$O(mn^{O(\varepsilon)}\operatorname{polylog}(n))$ extra time per iteration. This time is dominated by the runtime of the reduced oracle.

The total number of iterations will be $T'=T\cdot O(\log W)=T\cdot O(\log n)$, since we need
to run the algorithm inside the binary search over $\alpha\in[1,8W]$.
Each iteration succeeds with probability at least $1-n^{-\Theta(1)}$,
so all $T'$ iterations succeed with probability at least $1-T'n^{-\Theta(1)}\ge 1-o(1)$.

This gives the statements of Theorems~\ref{thm:main} and \ref{thm:main:padding},
except that the complexities have factors $n^{O(\varepsilon)}$ instead of $n^{\varepsilon}$.
Rescaling $\varepsilon$ by an appropriate constant gives the desired results.

\end{proof}
In Sections~\ref{sec:first-approach} and~\ref{sec:second-approach} below we describe the two implementations
of the reduced oracle stated in Theorem~\ref{thm:reduced-oracle}.
To simplify notation, we will assume that vectors $\tilde v_i$ for $i\in V$ are unique,
and rename the nodes in $V$ so that $\tilde v_x=x$ for each $x\in V$.
Thus, we now have $V\subseteq\mathbb R^d$. The ``true'' vector in $\mathbb R^n$ corresponding to $x\in V$ is still denoted as $v_x$.

For a path $p=(x,\ldots,y)$ we will denote $\partial p=\{x,y\}$ to be the pair of its endpoints.

We write $u\sim\calN(0,I_d)$ for a standard Gaussian vector in $\mathbb R^d$, whose components are independent $\calN(0,1)$ variables.
Notation ${\tt Pr}_u[\cdot]$ denotes probability over this choice of $u$.

\section{Maximum flow subroutine}\label{sec:flow-subroutine}
Before presenting the two approaches,
we describe a construction that will be used in both of them.

For disjoint subsets $A,B\subseteq S$ and parameter $\mu> 0$ let  $\calN(A,B,\mu)$ 
be the directed weighted graph with the nodes $\{x^+,x^-:x\in V\}\cup\{e^-,e^+:e\in E\}\cup\{s,t\}$
and the edges
$$
\bigcup_{x\in V}\{\underbrace{(x^+,x^-)}_{w(x)/2}\}
\cup
\!\!\!\bigcup_{e=xy\in E}\!\!\!\{\underbrace{(e^-,e^+)}_{w(e)},\underbrace{(x^-,e^-),(y^-,e^-),(e^+,x^+),(e^+,y^+)}_{\infty}\}
\cup
\bigcup_{x\in A}\{\underbrace{(s,x^+)}_{\mu}\}
\cup
\bigcup_{y\in B}\{\underbrace{(y^-,t)}_{\mu}\}
$$
(The numbers below the edges indicate their weights).

\begin{figure}[ht]
\begin{minipage}[c]{0.36\textwidth}
\caption{The network $\calN(\{x\},\{y\},\mu)$ for a graph consisting of one edge $e=xy$. Arc labels are capacities.}
\label{fig:single-edge-network}
\end{minipage}\hfill
\begin{minipage}[c]{0.61\textwidth}
\centering
\begin{tikzpicture}[>={Stealth[length=3mm,width=2.2mm]}, every node/.style={font=\small},
  vertex/.style={circle,draw,minimum size=7mm,inner sep=1pt}]
\node[vertex] (s) at (0,1.6) {$s$};
\node[vertex] (xp) at (-3,0) {$x^+$};
\node[vertex] (ep) at (0,0) {$e^+$};
\node[vertex] (yp) at (3,0) {$y^+$};
\node[vertex] (xm) at (-3,-2) {$x^-$};
\node[vertex] (em) at (0,-2) {$e^-$};
\node[vertex] (ym) at (3,-2) {$y^-$};
\node[vertex] (t) at (0,-3.6) {$t$};
\draw[->] (s) -- node[above left] {$\mu$} (xp);
\draw[->] (xp) -- node[left] {$w(x)/2$} (xm);
\draw[->] (xm) -- node[below] {$\infty$} (em);
\draw[->] (em) -- node[right] {$w(xy)$} (ep);
\draw[->] (ep) -- node[above] {$\infty$} (yp);
\draw[->] (yp) -- node[right] {$w(y)/2$} (ym);
\draw[->] (ym) -- node[below right] {$\mu$} (t);
\draw[->] (ep) -- node[above] {$\infty$} (xp);
\draw[->] (ym) -- node[below] {$\infty$} (em);
\end{tikzpicture}
\end{minipage}
\end{figure}

An $s$-$t$ cut $(S_{\tt cut},T_{\tt cut})$ in $\calN(A,B,\mu)$ will be called {\em valid} if there is no node $x\in V$ with $x^+\in T_{\tt cut},x^-\in S_{\tt cut}$.
An invalid cut can be converted to a valid one by moving $x^+$ to $S_{\tt cut}$ whenever $x^-\in S_{\tt cut}$; this transformation does not increase the cost of the cut.
Let us define a mapping from valid cuts $(S_{\tt cut},T_{\tt cut})$ to element cuts  $(X,Y,U)$
by setting
%partition $V=X\sqcup Y \sqcup U$ 
$X=\{x\in V:x^+,x^-\in S_{\tt cut}\}$,
$Y=\{x\in V:x^+,x^-\in T_{\tt cut}\}$,
$U=V-X-Y$.
The lemma below gives some structural properties of this construction.

\begin{lemma}\label{lem:flow-subroutine-routing}
(a)
For each element cut $(X,Y,U)$, the cost of a minimum valid cut $(S_{\tt cut},T_{\tt cut})$
that maps to $(X,Y,U)$  equals
$$
 \Phi(X,Y,U):=\frac{w(U)}2+w(E(X,Y))
                  +\mu|A\cap Y|+\mu|B\cap X|
                  \ge \frac 12 {\tt cost}(X,Y,U)
$$
In particular, if $(S_{\tt cut},T_{\tt cut})$ is a minimum valid cut in $\calN(A,B,\mu)$ 
and $(X,Y,U)$ is the corresponding element cut then 
$F:={\tt cost}(S_{\tt cut},T_{\tt cut})=\Phi(X,Y,U)=\min_{(X',Y',U')}\Phi(X',Y',U')$.
If in addition $|A|=|B|=k$ and $F\le \mu k/2$ then $(X,Y,U)$ is $k/(2n)$-balanced. \\
(b) 
$s$-$t$ flows in $\calN(A,B,\mu)$ are equivalent, with the same value, to nonnegative weightings $f_p$ of simple $A$-$B$ paths satisfying
\[
 \sum_{p:x\in p}f_p\le w(x)/2,\qquad
 \sum_{p:e\in p}f_p\le w(e),\qquad
 \sum_{p:x\in\partial p}f_p\le\mu
\]
for all nodes $x\in V$ and edges $e\in E$.
If $D$ is the demand graph of such $f$ and  $\lambda_e=\sum_{p:e\in p}f_p$
then
\[
 \deg(G^\lambda)\le w|_V,\qquad \lambda_e\le w(e),\qquad
 \sum_p f_pT_p-\sum_e\lambda_eL_e=-\calL(D),\qquad
 \|\calL(D)\|\le2\mu.
\]
\end{lemma}
\begin{proof}
\myparagraph{(a)}
The first claim can be checked mechanically. If $F\le\mu k/2$, then $|A\cap Y|,|B\cap X|\le F/\mu\le k/2$.
As $|A|=|B|=k$, at least $k/2$ vertices lie outside each of $X,Y$, proving balance.

\myparagraph{(b)}
Each simple $A$-$B$ path $p=(x_0,\ldots,x_k)$, with $e_i=x_{i-1}x_i$, lifts to the $s$-$t$ path
\[
 (s,x_0^+,x_0^-,e_1^-,e_1^+,x_1^+,x_1^-,\ldots,e_k^-,e_k^+,x_k^+,x_k^-,t)
\]
in $\calN(A,B,\mu)$.
Conversely, every simple $s$-$t$ path has this form: traversing an edge gadget back to the same vertex, or revisiting an original vertex, would repeat its split nodes.
Thus, removing flow cycles and decomposing the remaining flow into simple paths gives the inverse transformation without changing its value.
The split-arc, edge-bottleneck and source/sink capacities give exactly the vertex, edge and endpoint load bounds, respectively, proving the first claim and $\lambda_e\le w(e)$.
Each path uses at most two edges incident to any vertex, proving the degree bound.
The Laplacian identity follows by cancelling the edge Laplacians along each path, leaving its negative endpoint Laplacian.
Finally, the endpoint bound gives $\deg(D)\le\mu$ and hence $\|\calL(D)\|\le2\mu$.
\end{proof}
Graph $\calN(A,B,\mu)$ has $O(m)$ nodes and arcs. Note that infinite capacities can be replaced with $n\mu+1$ without affecting
Lemma~\ref{lem:flow-subroutine-routing} (since the total flow cannot exceed $n\mu$).

We will use graph $\calN(A,B,\mu)$ inside the following subroutine, which takes a vector in $\mathbb R^d$ as an input.
(Parameter $\mu$ is assumed to be fixed).

\begin{algorithm}[H]
  \DontPrintSemicolon
\SetNoFillComment
compute $h_x=\langle x,u\rangle$ for each $x\in S$ (``height'' of $x$ in the direction $u$) \\
sort $\{h_x\}_{x\in S}$ and let $A,B\subseteq S$, with $|A|=|B|=2c'n$, contain the smallest and largest heights, respectively, breaking ties by a fixed rule \\
construct graph $G'=\calN(A,B,\mu)$ \\
compute a maximum $s$-$t$ flow $f'$ in $G'$ and a minimum valid cut $(\Scut,\Tcut)$; let $F$ be their common value \\
remove cycles from $f'$ and implicitly decompose it into simple paths; project these to a weighted $A$-$B$ path routing $f$ in $G$ as in Lemma~\ref{lem:flow-subroutine-routing}(b) \\
compute demands $d_{xy}=\sum_{p:\partial p=\{x,y\}}f_p$ for $x\in A,y\in B$; let $D$ be the undirected graph with edge weights $d_{xy}$ \\
construct the element cut $(X,Y,U)$ from $(\Scut,\Tcut)$ as in Lemma~\ref{lem:flow-subroutine-routing}(a), and return $A,B,h,F,f,D,(X,Y,U)$
\caption{${\tt DirectionalFlowCut}(u)$.}\label{alg:DirectionalFlowCut}
\end{algorithm}

The routing $f$ is represented implicitly; $f'$ is internal to the subroutine and is a flow in $\calN(A,B,\mu)$.
The decomposition has $O(m)$ paths, but only their original endpoints and weights are needed to construct $D$.
Cycle removal and endpoint decomposition can be implemented in $O(m\log n)$ time using dynamic trees~\cite{SleatorTarjan}, without listing path interiors.
The implicit representation of $f$ also allows its edge loads $\sum_{p:e\in p}f_p$ to be read from the edge-bottleneck arcs of the cycle-free flow in $O(m)$ total time.
In particular, $D$ has $O(m)$ edges and total edge weight $F$; its symmetric adjacency matrix satisfies $D_{xy}=D_{yx}=d_{xy}$ for $x\in A,y\in B$.

\begin{lemma}[Gaussian direction]\label{lem:flow-subroutine-gaussian}
Under Assumption~\ref{assumption:main}, there are constants $c',\sigma,\zeta>0$, depending only on $c$, such that for $u\sim\calN(0,I_d)$ Algorithm~\ref{alg:DirectionalFlowCut} satisfies
\[
 \Pr{\min_{y\in B}\langle y,u\rangle-
             \max_{x\in A}\langle x,u\rangle\ge\sigma}\ge \zeta.
\]
For every direction, if $F\le c'n\mu$, the returned cut gives a $c'$-balanced element cut of cost at most $2c'n\mu$.
Otherwise the returned flow routes at least $c'n\mu$ units between $A$ and $B$ subject to the bounds of Lemma~\ref{lem:flow-subroutine-routing}; on the event above every demand pair has projection gap at least $\sigma$.
\end{lemma}
\begin{proof}
The bounded norms and spreading in Assumption~\ref{assumption:main} imply that $\Omega(n^2)$ ordered pairs in $S$ have squared distance bounded below by a positive constant.
For each such pair, $\langle y-x,u\rangle$ is a centered Gaussian with variance $\|y-x\|^2$, so it exceeds a suitable constant $\sigma>0$ with constant probability.
Thus the number $Z$ of ordered pairs with projection gap at least $\sigma$ has expectation $\Omega(n^2)$.
Since $0\le Z\le n^2$, there are constants $\theta,\zeta>0$ such that $Z\ge\theta n^2$ with probability at least $\zeta$.
If the gap between $A$ and $B$ were less than $\sigma$, every pair counted by $Z$ would have its lower endpoint in $A$ or its upper endpoint in $B$.
This would give $Z\le(|A|+|B|)|S|\le4c'n^2$, a contradiction for $4c'<\theta$.
The cut and flow conclusions follow from Lemma~\ref{lem:flow-subroutine-routing}.
\end{proof}

\section{First approach: rectangular covers}\label{sec:first-approach}
In this section
we assume that $\beta(G)\le\beta$,
or that $w(e)=\infty$ for all $e\in E$ and $\beta_v(G)\le\beta$
(see Definition~\ref{def:rect-padding}).

We start with some definitions.
A set $Q=A\times B$, where $A,B\subseteq S$ are disjoint, is a \emph{rectangle}.
We denote $\langle Q\rangle=\{\{x,y\}\::\:(x,y)\in (A\times B)\cup (B\times A)\}\subseteq \binom S2$.
Let $Q^{\tt paths}$ consist of the simple graph paths $p$ with $\partial p \in\langle Q\rangle$.
For a set $\calQ=\{Q_1,\ldots,Q_q\}$ with $q=|\calQ|$ rectangles define $\langle\calQ\rangle=\bigcup_i \langle Q_i\rangle\subseteq\binom S2$ and $\calQ^{\tt paths}=\bigcup_i Q_i^{\tt paths}$.

\begin{definition}\label{def:rect-cover}
A set $\calQ$ is a \emph{$(\varepsilon,\Delta)$-rectangular cover} if every pair  $\{x,y\}\in\langle\calQ\rangle$ 
satisfies $\|x-y\|^2\ge \varepsilon$, and every pair $\{x,y\}\in\binom S2$ with $\|x-y\|^2\ge \Delta$ belongs to $\langle\calQ\rangle$.
A nonnegative vector $\nu:\calQ^{\tt paths}\to\mathbb R_{\ge0}$ is a \emph{$\mu$-uniform $\calQ$-flow} if
\begin{subequations}
\begin{align}
 \sum_{p:x\in p}\nu(p)&\le w(x)/2 &&(x\in V)\label{eq:rect-vertex-cap}\\
 \sum_{p:e\in p}\nu(p)&\le w(e) &&(e\in E)\label{eq:rect-edge-cap} \\
 \sum_{p:x\in\partial p}\nu(p)&\le\mu &&(x\in S)\label{eq:rect-endpoint-cap}
\end{align}
\end{subequations}
Its value is ${\tt val}(\nu)=\sum_p\nu(p)$. The {\em demand graph} $D_\nu$ is an undirected weighted graph on $V$
with positive weights where the weight of edge $xy$ equals $\sum_{p:\partial p=\{x,y\}}\nu(p)$.
\end{definition}

Our algorithm will work as follows. First, we will generate a set of rectangles $\calQ$ by sampling Gaussian directions
which w.h.p.\ will be an $(\varepsilon,\Delta)$-rectangular cover for $\Delta=\Theta(1)$ (Sec.~\ref{sec:rectangular-cover}).
Next, we will approximately solve the maximum $\mu$-uniform $\calQ$-flow problem (Sec.~\ref{sec:Q-flow-problem}). If this flow has a large value
then we will use it to output a valid feedback matrix (Sec.~\ref{sec:Q-flow-problem:large}).
Otherwise, calling ${\tt DirectionalFlowCut}(u)$ for $u\sim \calN(0,I_d)$ will give a balanced element cut with a sufficiently high probability (Sec.~\ref{sec:Q-flow-problem:small}).

\subsection{Generating rectangular cover}\label{sec:rectangular-cover}
The construction will be based on the following lemma.
\begin{lemma}\label{lem:rect-one-direction}
Consider the following procedure for parameter $\sigma=\sqrt{6\varepsilon \ln n}$:
\begin{itemize}
\item Sample $u\sim\calN(0,I_d)$.
\item For an integer $k$ and operator $\circ\in\{\le,\ge\}$ define $S_{\circ k}=\{x\in S\::\:\langle x,u\rangle\circ k\sigma\}$. 
\item Output set of rectangles $\calQ=\{S_{\le k}\times S_{\ge k+1}\::\:k\in\mathbb Z,\mbox{~$S_{\le k}$ and $S_{\ge k+1}$ are non-empty}\}$.
\end{itemize}
(a) With probability at least $1-\tfrac 1n$, every $\{x,y\}\in \langle\calQ\rangle$ satisfies
$ \|x-y\|^2\ge \varepsilon$. \\
(b) With probability at least $1-\tfrac 1n$,
$
 |\calQ|\le O\left(\frac{\sqrt{\ln n}}\sigma\right) = O\left(\tfrac 1 {\sqrt\varepsilon}\right)
$.
\\
(c) Every $\{x,y\}\in\binom S2$ with $\|x-y\|^2\ge\Delta$ satisfies
$
 \Pr{\{x,y\}\in\langle\calQ\rangle}\ge e^{-O((\sigma/\sqrt{\Delta}+1)^2)}
$.
In particular, if $\Delta=\Theta(1)$ and $\sigma=\Omega(1)$ then this lower bound is $e^{-O(\sigma^2)}=n^{-O(\varepsilon)}$.
\end{lemma}
\begin{proof}
For $x\in\mathbb R^d$ write $h_x=\langle x,u\rangle$.
Note that for any distinct $x,y\in\mathbb R^d$, the quantity $h_x-h_y=\langle x-y,u\rangle$ is normal with zero mean and variance $\|x-y\|^2$.
% and let $Z\sim\calN(0,1)$.

\myparagraph{(a)} Consider pair $\{x,y\}\in\binom S 2$ with $r^2=\|x-y\|^2<\varepsilon=\sigma^2/(6\ln n)$.
If $\{x,y\}\in\langle\calQ\rangle$ then $|h_x-h_y|\ge\sigma$.
Therefore,
$$
 \Pr{\{x,y\}\in\langle\calQ\rangle}  \le \Pr{|\calN(0,r^2)|\ge\sigma}
 =2\Pr{\calN(0,1)\ge \sigma/r}
 \le 2e^{-\sigma^2/(2r^2)}\le2n^{-3}
$$
A union bound over at most $\binom n2$ pairs proves the claim.

\myparagraph{(b)}
For each $x\in S$ we have $\|x\|^2\le r^2\eqdef 4/c$ by Assumption~\ref{assumption:main}. Therefore,
$$
 \Pr{|h_x|\ge \Theta(\sqrt{\ln n})} \le \Pr{|\calN(0,r^2)|\ge \Theta(\sqrt{\ln n})}
 \le n^{-2}
$$
for an sufficiently large constant in the $\Theta(\cdot)$ notation.
By the union bound, with probability at least $1-\tfrac 1n$ we have $|h_x|\le \Theta(\sqrt{\ln n})$ for all $x\in S$.
It remains to notice that $|\calQ|\le O(\max_{x\in S} |h_x|/\sigma)$ by construction.

\myparagraph{(c)}
Consider $\{x,y\}\in \binom S 2$ with $\|x-y\|^2\ge \Delta$.
By construction, if $|h_x-h_y|\ge 2\sigma$  then $\{x,y\}\in\langle\calQ\rangle$.
Therefore,
$$
 \Pr{\{x,y\}\in\langle\calQ\rangle}
 \ge\Pr{|\calN(0,\Delta)| \ge 2 \sigma}
 =2\Pr{\calN(0,1) \ge 2 \sigma/\sqrt\Delta}
 \ge e^{-O((\sigma/\sqrt{\Delta}+1)^2)}
$$
\end{proof}

\begin{corollary}\label{cor:rect-cover}
Running the procedure in Lemma~\ref{lem:rect-one-direction} independently $N=n^{\Theta(\varepsilon)}\log n$ times and taking the union
gives set $\calQ$ which with probability at least $1-n^{-1/2}$ is a $(\varepsilon,\Delta)$-rectangular cover and has $q=|\calQ|\le n^{O(\varepsilon)}\log^{3/2}n$ rectangles,
assuming that $\Delta=\Theta(1)$ and $\varepsilon\in[\Theta(1/\log n),\Theta(1)]$ with appropriate constants.
\end{corollary}
\begin{proof}
The probability that a given $\{x,y\}\in\binom S 2$ with $\|x-y\|^2\ge \Delta$ will {\bf not} be included in any $\langle\calQ\rangle$
is at most $(1-n^{-O(\varepsilon)})^{N}\le \tfrac 2 {n^3}$ with an appropriate choice of constants.
The number of such pairs is at most $\binom n 2$. By the union bound, the output $\calQ$ will have desired properties 
with probability at least $1-N\cdot \tfrac 2n-\binom n 2\cdot \tfrac{2}{n^3}\ge 1-n^{-1/2}$.
\end{proof}

From now on we assume that we a given a $(\varepsilon,\Delta)$-rectangular cover $\calQ=\{Q_1,\ldots,Q_q\}$ with $q\le n^{O(\varepsilon)}\log^{3/2}n$ rectangles.
We denote ${\tt OPT}_\calQ$ to be the maximum value of the $\mu$-uniform $\calQ$-flow, where constant $\mu$ will be specified later.

\subsection{Solving the maximum $\mu$-uniform $\calQ$-flow problem}\label{sec:Q-flow-problem}
We will reduce the problem to a maximum multicommodity flow problem in a directed
weighted graph $\calG$ and then apply the algorithm of Chen-Graur-Sidford~\cite{ChenGraurSidford2025} for the latter problem.
Start with the network $\calN(\varnothing,\varnothing,\mu)$ of Section~\ref{sec:flow-subroutine}, remove $s,t$, and add:
\begin{itemize}
\item for each $x\in S$, nodes $x^\ominus,x^\oplus$, an arc $(x^\ominus, x^\oplus)$ of capacity $\mu$, and infinite-capacity arcs $(x^\oplus, x^+)$ and $(x^-, x^\ominus)$;
\item for each rectangle $Q_i=A_i\times B_i$, private nodes $s_i,t_i$ and infinite-capacity arcs $(s_i, x^\ominus)$ for $x\in A_i$ and $(y^\oplus, t_i)$ for $y\in B_i$.
\end{itemize}
Note that $\calG$ has $M=O(mq)$ nodes and edges.

Now consider the maximum multicommodity flow problem in $\calG$ with $q$ terminal pairs $\{(s_i,t_i)\}_{i=1}^q$.
We claim that it is equivalent to the maximum $\mu$-uniform $\calQ$-flow problem,
i.e.\ there are value-preserving transformations between the problems in each direction.
Indeed, use the path lifting and projection from Lemma~\ref{lem:flow-subroutine-routing}(b) for the vertex and edge gadgets. For a path with endpoints $a\in A_i,b\in B_i$,
prepend $(s_i,a^\ominus,a^\oplus)$ and append $(b^\ominus,b^\oplus,t_i)$.
The gadget arcs enforce exactly the vertex and edge loads, and each endpoint gate is used once per incident demand, enforcing the shared bound $\mu$ across all rectangles and both endpoint roles.
Conversely, remove flow cycles and project each simple $s_i$-$t_i$ path.
Such a path cannot traverse an endpoint gate internally, since it would revisit the corresponding split vertex; private sources and sinks cannot occur internally either.
The resulting paths therefore lie in $Q_i^{\tt paths}$ and satisfy all three load constraints with the same total value.

Next, we apply the algorithm in~\cite{ChenGraurSidford2025} for the approximate multicommodity flow problem.
It produces a flow of value at least $(1-\varepsilon')$ of the optimum in $O(M^{1+o(1)}q/\varepsilon')$  time,
assuming that weights are polynomially-bounded integers.
(In our case edges with infinite weights can be safely replaced with $Mc_{\max}+1$ where $c_{\max}$ is the maximum finite weight in $\calG$).
Cycle removal and endpoint decomposition take $O(M\log(n+q))$ time per commodity using dynamic trees~\cite[Section 6]{SleatorTarjan} and pseudo-decomposition~\cite[Lemma 3.7]{Sherman}.
Recording only the endpoint pairs and their weights constructs the demand graph $D_\nu$ with $O(Mq)$ edges in total, without listing path interiors.

We arrive at the following result.
\begin{theorem}\label{th:rect-mcf}
For a given value $\mu=a/b$ with polynomially-bounded integers $a,b$ we can in $m^{1+o(1)}q^{2+o(1)}$ time compute
a $\mu$-uniform $\calQ$-flow $\nu$, represented implicitly, of value at least ${\tt OPT}_\calQ/2$, and construct the demand graph $D_\nu$ with $O(mq^2)$ edges.
\end{theorem}

\begin{remark}
Alternatively, an approximate $\mu$-uniform $\calQ$-flow with desired
guarantees could be computed via calls to a minimum-cost flow subroutine using the MWU frawework;
this was the solution originally proposed by AI. We replaced it with a reduction to the algorithm in~\cite{ChenGraurSidford2025},
which is a shorter way to get the desired complexity.
\end{remark}
%We will use
%\begin{equation}\label{eq:rect-parameters}
% F=\frac{4\alpha}{\varepsilon},\qquad
% \mu=\frac{CF\lceil\sqrt\beta\rceil}{n},
%\end{equation}
%where $C$ is a sufficiently large constant, depending only on $c$ and the absolute padding constant.

\subsection{Feedback matrix for a large flow (${\tt val}(\nu)\ge 2\alpha/\varepsilon$) }\label{sec:Q-flow-problem:large}
Let $\nu$ be the flow from Theorem~\ref{th:rect-mcf}. In this section we show the following result.
% for the reduced oracle specification of Section~\ref{sec:low-rank}.
\begin{lemma}\label{lem:rect-feedback}
Suppose that parameters $\gamma,\tau$ in eq.~\eqref{eq:GramApproximation}
satisfy $\tau\le\xi/2$ and
$\gamma\le\min\{1/4,\varepsilon/(2(\varepsilon+\tau))\}$.
If ${\tt val}(\nu)\ge2\alpha/\varepsilon$
then $N=-\calL(D_\nu)$ is a valid feedback matrix of the reduced oracle with $\|N\|\le 2\mu$.
This matrix has $O(mq^2)$ non-zero entries.
\end{lemma}
\begin{proof}
Set $\lambda_e=\sum_{p:e\in p}\nu(p)$.
The load and Laplacian calculations in the proof of Lemma~\ref{lem:flow-subroutine-routing}(b) apply to any path routing with these vertex, edge and endpoint bounds, including $\nu$.
Thus $N=-\calL(D_\nu)$ has the required form and satisfies $\|N\|\le2\mu$.

Every edge $xy$ of graph $D_\nu$ belongs to $\langle\calQ\rangle$ and thus satisfies $d_{xy}=\|x-y\|^2\ge\varepsilon$.
By~\eqref{eq:GramApproximation} and the choice of $\gamma$,
\[
 \|v_x-v_y\|^2\ge(1-\gamma)d_{xy}-\gamma\tau
 \ge\varepsilon-\gamma(\varepsilon+\tau)\ge\varepsilon/2.
\]
Consequently,
\[
 N\bullet X=-\sum_{xy\in E(D_\nu)}(D_\nu)_{xy}\|v_x-v_y\|^2
 \le-\frac{\varepsilon}{2}{\tt val}(\nu)\le-\alpha
\]

\end{proof}

%\subsection{Lipschitz approximation of Gaussian projections}\label{sec:gaussian-lipschitz}
\subsection{Randomized sampling procedure for a small flow (${\tt val}(\nu)<2\alpha/\varepsilon$)}\label{sec:Q-flow-problem:small}
Assume now that ${\tt val}(\nu)<2\alpha/\varepsilon$;
this implies that ${\tt OPT}_{\calQ}<4\alpha/\varepsilon$.
We will prove the following result.

\begin{lemma}\label{lem:rect-cheap-cut}
There are constants $c',\zeta,\Delta>0$ such that
 a call to ${\tt DirectionalFlowCut}(u)$ with $\mu\ge\Theta(\frac{\alpha\sqrt\beta}{\varepsilon n})$ for an appropriate constant and $u\sim\calN(0,I_d)$
 has minimum network-cut value $F<c'n\mu$ and produces a $c'$-balanced element cut $(X,Y,U)$ of cost at most $2c'n\mu$ with probability at least $\zeta/2$.
\end{lemma}

The lemma will imply $\Theta(\ln n)$ independent calls to  ${\tt DirectionalFlowCut}(u)$ will give a desired element cut with probability at least $1-\tfrac 1n$.

%Since the rectangular cover contains every pair at squared distance at least $\Delta$, the dual constraints apply to all such pairs.
Fix some $\mu>0$. Let $(\eta,\xi)\in\mathbb R^{V\cup E}_{\ge 0}\times\mathbb R^{S}_{\ge 0}$ be an optimal solution of the dual problem, then
%\begin{equation}\label{eq:rect-dual-prices}
$
 \sum_{x\in V}\frac{w(x)}2\eta_x+\sum_{xy\in E}w(xy)\eta_{xy}+\mu\sum_{x\in S}\xi_x<4\alpha/\varepsilon
 $
%\end{equation}
and
\begin{equation}\label{eq:rect-path-prices}
 \sum_{x'\in p}\eta_{x'}+\sum_{x'y'\in p}\eta_{x'y'}+\xi_x+\xi_y\ge 1\qquad\forall p\in \calQ^{\tt paths}, \partial p = \{x,y\}
\end{equation}
Let us modify $\eta_x,\eta_{xy},\xi_x$ as follows: (i) truncate all values $\xi_x$ with $\xi_x < 1/3$ to $0$;
(ii) multiply all variables $\eta_x,\eta_{xy},\xi_x$ by $3$;
(iii) truncate all values $\xi_x$ with $\xi_x > 1$ to $1$. Clearly,~\eqref{eq:rect-path-prices} still holds, and the new variables
satisfy
\begin{equation}\label{eq:rect-dual-prices}
J:= \sum_{x\in V}\frac{w(x)}2\eta_x+\sum_{xy\in E}w(xy)\eta_{xy}+\mu\sum_{x\in S}\xi_x<12\alpha/\varepsilon
\end{equation}
Furthermore, $\xi_x\in\{0,1\}$ for each $x\in S$.
Define
\[
 S^*:=\{x\in S:\xi_x=0\},\qquad
 \xi_{xy}:=\eta_x+\eta_{xy}+\eta_y\quad(xy\in E).
\]
We have 
\begin{equation}\label{eq:rect-priced-terminals}
 \mu\cdot|S- S^*|\le J<\frac{12\alpha}{\varepsilon}.
\end{equation}

Let $\varrho$ be the shortest-path pseudometric on $V$ with edge lengths $\xi_{xy}$.
Consider $x,y\in S^*$ with $\|x-y\|^2\ge\Delta$. Every path $p$ with $\partial p=\{x,y\}$ belongs to $\calQ^{\tt paths}$,
and thus satisfies
\begin{equation}\label{eq:rect-path-separation}
 \sum_{x'y'\in p}\xi_{x'y'}\ge  \sum_{x'\in p}\eta_{x'}+\sum_{x'y'\in p}\eta_{x'y'}+0+0 \ge 1
\end{equation}
by eq.~\eqref{eq:rect-path-prices}.
Therefore,
\begin{equation}\label{eq:rect-metric-implication}
 x,y\in S^*,\quad\|x-y\|^2\ge\Delta
 \quad\Longrightarrow\quad\varrho(x,y)\ge1,
\end{equation}
while $\varrho(x,y)\le\xi_{xy}$ for every edge $xy\in E$.
%By Assumption~\ref{assumption:main}, the terminal vectors have norm at most $2/\sqrt c$.

Next, we prove a useful property of the metric $\varrho$ constructed in this way.
Below $\theta$ is assumed to be a parameter in $(0,1)$.

\begin{sloppypar}
\begin{lemma}[Lipschitz approximation of Gaussian projections]\label{lem:gaussian-lipschitz}
There exists function $L=L(c,\Delta,\theta,\Gamma,\beta)$ with $L=\Theta_{c,\Delta,\theta,\Gamma}(\sqrt\beta)$
with the following property: for a given set of parameters $c,\Delta,\theta,\Gamma,\beta,\mu$ there exists a set $\widehat S\subseteq S^*$ with
$|\widehat S|\ge(1-\theta)|S^*|$ such that, for a standard Gaussian vector $u\sim\calN(0,I_d)$, there is an $L$-Lipschitz function
$H_u:(V,\varrho)\to\mathbb R$ %, where $L=\Theta(\sqrt\beta)$, 
satisfying
\begin{equation}\label{eq:rect-projection-error}
 \mathbb E_u\bigl|H_u(x)-\langle u,x\rangle\bigr|
 \le2\sqrt\Delta\qquad(x\in \widehat S).
\end{equation}
%The set $\widehat S$ and the implicit constant in $L=\Theta(\sqrt\beta)$ may depend on $c,\Delta,\theta,\Gamma$ but not on $u,\mu$.
In particular, if $w(e)=\infty$ for all $e\in E$, it suffices that $\beta_v(G)\le\beta$.
\end{lemma}
\end{sloppypar}

\begin{proof}
Note first that if $w(e)=\infty$ for all $e\in E$, then $\eta_{xy}=0$ for all $xy\in E$ since~\eqref{eq:rect-dual-prices}
is finite (and hence $\varrho$ has edge lengths $\xi_{xy}=\eta_x+\eta_y$).

Put $r=\Gamma/4$.
Apply the padding bound to $(V,\varrho)$ with diameter bound $D=1/2$.
Call a terminal padded if its entire $\varrho$-ball of radius $2r$ lies in its part.
The padding bound gives this event probability at least $e^{-\Gamma\beta}$, since $2r=\Gamma/2$; using radius $r$ for the cores below leaves positive slack.
Take $M=\lceil e^{\Gamma\beta}\ln(1/\theta)\rceil$ independent partitions.
Each terminal is unpadded in all of them with probability at most
$(1-e^{-\Gamma\beta})^M\le\theta$.
Hence there exist $M$ partitions of $V$ such that at least $(1-\theta)|S^*|$ terminals of $S^*$ are padded in at least one partition. Denote these fixed partitions by $\mathcal P_1,\ldots,\mathcal P_M$, and let $\widehat S$ be the set of such terminals.

Set
\[
 L=\frac{2\sqrt{8\ln(16M/(c\Delta))}}{r\sqrt c}
   =\Theta(\sqrt\beta),\qquad b=Lr/2.
\]
For every part $P\in\bigcup_{j=1}^M\mathcal P_j$, define its {\em core} by
\[
 P^{\tt core}=\{x\in P\cap S^*:\varrho(x,V-P)\ge r\}
\]
Let
\[
 \Pi:=\left\{P\in\bigcup_{j=1}^M\mathcal P_j:P^{\tt core}\ne\varnothing\right\},
\]
and choose a representative $x_P\in P^{\tt core}$ for each $P\in\Pi$.
Since each part has $\varrho$-diameter at most $1/2$,~\eqref{eq:rect-metric-implication} gives
\[
 \|x-x_P\|<\sqrt\Delta\qquad(P\in\Pi,\ x\in P\cap S^*).
\]
Let $h_P$ be $\langle u,x_P\rangle$ clipped to $[-b,b]$, and define
\begin{equation}\label{eq:Hux:def}
 H_u(x)=\frac12\min_{P\in\Pi}\{h_P+L\varrho(x,P^{\tt core})\} 
        +\frac12\max_{P\in\Pi}\{h_P-L\varrho(x,P^{\tt core})\}.
\end{equation}
Both envelopes defining $H_u$ are $L$-Lipschitz, and thus so is their midpoint $H_u$.

Fix $x\in \widehat S$ and let $\Pi_x=\{P\in\Pi:x\in P\}$.
Each of the $M$ partitions has exactly one part containing $x$, and thus $1\le|\Pi_x|\le M$.

We claim that each extremum in~\eqref{eq:Hux:def} is attained at a possibly different part $P\in\Pi_x$.
%We claim that both extrema in~\eqref{eq:Hux:def} may be restricted to $\Pi_x$.
Indeed, since $x$ is padded in some partition, there is $\tilde P\in\Pi_x$ with $x\in\tilde P^{\tt core}$.
We have $h_{\tilde P}\pm L\varrho(x,\tilde P^{\tt core})=h_{\tilde P}\pm 0\in[-b,b]$ for this $\tilde P$.
Now consider $P\in\Pi-\Pi_x$.
Pick $y\in P^{\tt core}$. We have $x\in V-P$, so
$\varrho(x,y)\ge\varrho(y,V-P)\ge r$ by the definition of a core.
Therefore,
$h_P+L\varrho(x,P^{\tt core})\ge-b+Lr=b$ and $h_P-L\varrho(x,P^{\tt core})\le b-Lr=-b$. The claim follows.

Put $d_P=L\varrho(x,P^{\tt core})$ and define the unclipped local midpoint
\[
 \widetilde H_u(x)=\frac12\min_{P\in\Pi_x}\{\langle u,x_P\rangle+d_P\}
             +\frac12\max_{P\in\Pi_x}\{\langle u,x_P\rangle-d_P\}.
\]
We have
\[
 |H_u(x)-\langle u,x\rangle|
 \le \underbrace{|\widetilde H_u(x)-\langle u,x\rangle|}_{A(u)}
     +\underbrace{\max_{P\in\Pi_x}(|\langle u,x_P\rangle|-b)_+}_{B(u)}.
\]
Indeed, clipping changes each height by at most $B(u)$, hence changes each local extremum, and their midpoint, by at most $B(u)$.
We bound the two terms separately.
\begin{itemize}
\item Let us define
\begin{align*}
Z(u)&=\widetilde H_u(x)-\langle u,x\rangle=
 \tfrac 12\min\limits_{P\in\Pi_x}\{\langle u,x_P-x\rangle+d_P\}
      +\tfrac 12\max\limits_{P\in\Pi_x}\{\langle u,x_P-x\rangle-d_P\}
\\
 R(u)&=\max_{P\in\Pi_x}\{\langle u,x_P-x\rangle-d_P\}
\end{align*}
Clearly, we have $Z(u)=(R(u)-R(-u))/2$.
Recall that $\|x_P-x\|<\sqrt\Delta$ for every $P\in\Pi_x$.
Consequently, $R$ and $Z$ are $\sqrt\Delta$-Lipschitz as functions of $u$.
Moreover, $Z$ is odd, so $\mathbb E Z(u)=0$.
The Gaussian Poincar\'e inequality and Cauchy-Schwarz give
\[
 \mathbb E A(u)\le\sqrt{\mathbb E Z(u)^2}
 =\sqrt{\operatorname{Var}(Z(u))}\le\sqrt\Delta.
\]
\item Recall that $|\Pi_x|\le M$ and $\|x_P\|\le2/\sqrt c$ (by Assumption~\ref{assumption:main}).
The Gaussian tail bound and integration give
\[
 \mathbb E B(u)^2\le\frac{16M}{c}e^{-cb^2/8}=\Delta,
 \qquad \mathbb E B(u)\le\sqrt\Delta.
\]
\end{itemize}
Adding the two bounds yields $\mathbb E|H_u(x)-\langle u,x\rangle|\le2\sqrt\Delta$, proving~\eqref{eq:rect-projection-error}.
\end{proof}

We are ready to prove Lemma~\ref{lem:rect-cheap-cut}.
\begin{proof}[Proof of Lemma~\ref{lem:rect-cheap-cut}]
Take $c',\sigma,\zeta$ from Lemma~\ref{lem:flow-subroutine-gaussian}.
Let $I\subset\mathbb R$ be the projection gap between $A$ and $B$ trimmed by $\sigma/4$ at both ends
(or empty if the gap is at most $\sigma/2$). By the lemma, $\Pr{|I|\ge \sigma/2}\ge \zeta$.
Choose constants $\theta\in(0,1)$, $\Delta>0$ and parameter $\mu$ 
so that
%, apply
%Lemma~\ref{lem:gaussian-lipschitz}, and then choose $\mu$ so that
\[
 \theta+\frac{16\sqrt\Delta}{\sigma}\le\frac{c'\zeta}{4}
 \qquad\quad
 \mu \ge \frac{24\alpha}{c'\varepsilon n}\left(\frac{4L}{\sigma}+1\right)
\]
where $L=L(c,\Delta,\theta,\Gamma,\beta)$ is the function from Lemma~\ref{lem:gaussian-lipschitz}.
Note that $\mu\ge\Theta(\frac{\alpha\sqrt\beta}{\varepsilon n})$ since 
$L=\Theta(\sqrt\beta)$ and $\beta\ge 1$.
By Lemma~\ref{lem:gaussian-lipschitz} and Markov's inequality, for each
$x\in\widehat S$ we have
$\Pr{|H_u(x)-\langle u,x\rangle|>\sigma/8}\le16\sqrt\Delta/\sigma$.
Call a terminal in $S^*$ {\em bad} if it lies outside $\widehat S$ or its approximation error exceeds $\sigma/8$.
Let $N^{\tt bad}$ be the number of bad terminals.
Their expected number is at most
$
 \left(\theta+\frac{16\sqrt\Delta}{\sigma}\right)n
 \le\frac{c'\zeta n}{4}
$, and thus $\Pr{N^{\tt bad}>c'n/2}\le\zeta/2$ by
 Markov's inequality. By the union bound, we have $N^{\tt bad}\le c'n/2$ and $|I|\ge \sigma/2$ with probability at least
 $\zeta-\zeta/2=\zeta/2$.
 We will show that under this event the cost of the cut $F$ in the call ${\tt DirectionalFlowCut(u)}$ satisfies
 $F\le c'n\mu$; Lemma~\ref{lem:flow-subroutine-routing}(a) will then imply the desired claim (recall that $k=|A|=|B|=2c'n$).

On the event above, for $s\in I$ define
\begin{align*}
 A_s&=\{x\in V\::\:H_u(x)<s-L\eta_x\}\\
 B_s&=\{x\in V\::\:H_u(x)>s+L\eta_x\}\\
 C_s&=\{x\in V\::\:H_u(x)\in[s-L\eta_x,s+L\eta_x]\}
\end{align*}
If a good terminal $x\in A$, then
$\langle u,x\rangle\le s-\sigma/4$ and hence
$H_u(x)\le s-\sigma/8$; symmetrically, a good terminal $y\in B$ satisfies
$H_u(y)\ge s+\sigma/8$.  Thus every good terminal in $A\cup B$ either belongs
to $C_s$ or lies on its correct height shore.
The triple $(A_s,B_s,C_s)$ is an element cut of $G$.  By
Lemma~\ref{lem:flow-subroutine-routing}(a), there is a valid
$s$-$t$ cut in the graph $\calN(A,B,\mu)$ that maps to this element
cut and has capacity at most
\begin{align*}
 \Phi_s:&=\underbrace{\frac{w(C_s)}2}_{\sum\limits_{x\in V}[|s-H_u(x)|\le L\eta_x]\cdot \frac{w(x)}2}
       +\underbrace{w(E(A_s,B_s))}_{\sum\limits_{(x,y)\in\vec E}[s\in(H_u(x)+L\eta_x,H_u(y)-L\eta_y)]\cdot w(xy)}
       +\underbrace{\mu|A\cap B_s|+\mu|B\cap A_s|}_{\Psi},
\end{align*}
where $\vec E$ contains one orientation $(x,y)$ of every edge $xy\in E$, chosen so that $H_u(x)\le H_u(y)$.
Next, we will upper-bound $\int_{s\in I}\Phi_sds$.
Node $x\in V$ contributes at most $2L\eta_x\cdot\frac{w(x)}2=L\eta_xw(x)$ to this bound,
while edge $(x,y)\in\vec E$ contributes at most
\[
 \bigl(|H_u(x)-H_u(y)|-L\eta_x-L\eta_y\bigr)_+\cdot w(xy)
 \le
 \bigl(L\hspace{-20pt}\underbrace{\varrho(x,y)}_{\le\xi_{xy}=\eta_x+\eta_{xy}+\eta_y}\hspace{-20pt}-L\eta_x-L\eta_y\bigr)_+\cdot w(xy)
 \le L\eta_{xy}w(xy)
\]
where we used $L$-Lipschitzness of $H_u$. We conclude that
\[
 \int_{s\in I}\Phi_sds
 \le L\sum_{x\in V}\eta_xw(x)+L\sum_{xy\in E}\eta_{xy}w(xy)
      +|I|\Psi
 \le 2LJ+|I|\Psi
\]
Note that
$
\Psi\le\mu\left(|S-S^\ast|+N_{\tt bad}\right) \le\mu|S-S^*|+\frac{c'n\mu}{2}\le J+\frac{c'n\mu}{2}
$. We can now obtain the desired bound of $F$ as follows:
\begin{align*}
F\le\min_{s\in I}\Phi_s&\le \frac{2LJ}{|I|}+J+\frac{c'n\mu}{2}
\le \left(\frac{2L}{\sigma/2}+1\right)\cdot 12\alpha/\varepsilon+\frac{c'n\mu}{2}
\le \frac{c'n\mu}{2}+\frac{c'n\mu}{2}\le c'n\mu
\end{align*}

\end{proof}

\subsection{Wrap-up: proof of the first oracle bound}

We now collect the complexity, width, and sparsity bounds for the first
implementation of the reduced oracle.

\begin{lemma}\label{lemma:rect-oracle-bounds}
Let $\varepsilon\in[\Theta(1/\log n),\Theta(1)]$, and suppose that the
parameters $\Gamma$ and $\beta\in[1,O(\sqrt{\log n})]$ of the padding bound
are known.  The oracle described above runs in
$m^{1+o(1)}n^{O(\varepsilon)}\operatorname{polylog}(n)$ time and produces
either a $\Theta(1)$-balanced element cut of cost
$O(\alpha\sqrt\beta/\varepsilon)$ or a valid feedback matrix with spectral
norm $O(\alpha\sqrt\beta/(\varepsilon n))$, with probability at least $1-n^{-\Theta(1)}$.
Moreover, every feedback matrix has
$O(mn^{O(\varepsilon)}\operatorname{polylog}(n))$ non-zero entries.
\end{lemma}

\begin{proof}
Choose parameters $\Delta=\Theta(1)$ and $ \mu=\Theta\left(\frac{\alpha\sqrt\beta}{\varepsilon n}\right)$
as described above. Note that $\mu$ can be chosen
as a ratio $a/b$ of two polynomially-bounded integers within a constant factor of the required value, since $\frac{\alpha\sqrt\beta}{\varepsilon n}\in[n^{-\Theta(1)},n^{\Theta(1)}]$.
The graphs constructed inside the oracle may have non-integral rational weights of the form $\mu$ and $\frac{w(x)}2$;
they can be converted to polynomially-bounded integers by multiplying all weights by $2b$,
and then scaling the result back.

By Corollary~\ref{cor:rect-cover} and Theorem~\ref{th:rect-mcf}
a rectangular cover $\calQ$ with $q=n^{O(\varepsilon)}\operatorname{polylog}(n)$ rectangles and a corresponding $\calQ$-flow can be constructed in
$m^{1+o(1)}q^{2+o(1)}
=m^{1+o(1)}n^{O(\varepsilon)}\operatorname{polylog}(n)$ time,
with probability of failure at most $n^{-\Theta(1)}$.
In the case of a large flow 
Lemma~\ref{lem:rect-feedback} gives a feedback matrix $N$ with
$\|N\|\le 2\mu=O(\alpha\sqrt\beta/(\varepsilon n))$ and at most
$O(mq^2)=O(mn^{O(\varepsilon)}\operatorname{polylog}(n))$ non-zero entries.  
In the case of a small flow calling procedure in Lemma~\ref{lem:rect-cheap-cut}
$\Theta(\log n)$ times gives $\Theta(1)$-balanced cut of cost at most $O(\alpha\sqrt\beta/\varepsilon)$
with probability at least $1-n^{-\Theta(1)}$.

Taking $\tau=\Theta(1)$ sufficiently small and
$\gamma=\Theta(\varepsilon)$ meets Lemma~\ref{lem:rect-feedback}, and the
assumed range of $\varepsilon$ gives $\gamma=\Omega(1/\log n)$.

\end{proof}

\section{Second approach: chaining with correlated Gaussians}\label{sec:second-approach}
%\subsection{Procedure \texorpdfstring{${\tt Matching}(u)$}{Matching(u)}}
In this section we describe an alternative implementation of the reduced oracle. It follows very closely the oracle in~\cite{vnk:sparsest-cut}.

The main building block of the oracle is a procedure that takes vector $u\in\mathbb R^d$ and either outputs a directed matching $M$ on nodes $V$
or terminates the oracle. It works as follows ($c',\Delta,\sigma,\mu$ are positive constants that will be specified later):

\begin{algorithm}[H]
  \DontPrintSemicolon
\SetNoFillComment

	call ${\tt DirectionalFlowCut}(u)$ to obtain $A,B$, heights $h_x$, value $F$, routing $f$, demand graph $D$ with weights $d_{xy}$, and element cut $(X,Y,U)$ \\
	if $F\le c'n\mu$ then return $(X,Y,U)$ and terminate the oracle; by Lemma~\ref{lem:flow-subroutine-routing}(a), this cut is $c'$-balanced and has cost at most $2c'n\mu$ \\
	if $\sum_{x\in A,y\in B}d_{xy}\|x-y\|^2\ge2\alpha$, then set $\lambda_e=\sum_{p:e\in p}f_p$ for each $e\in E$, return $\lambda$ and $N=-\calL(D)$, and terminate the oracle \\
	let $M_{\tt all}=\{(x,y)\in A\times B\::\:d_{xy}>0,h_y-h_x\ge \sigma\}$ and $M_{\tt short}=\{(x,y)\in M_{\tt all}\::\:\|x-y\|^2\le \Delta\}$ \\
	pick maximal matching $M\subseteq M_{\tt short}$ and return $M$
      \caption{${\tt Matching}(u)$. 
      }\label{alg:Matching}
\end{algorithm}

\begin{lemma}\label{lemma:oracle7}
Suppose that parameters $\tau,\gamma$ in \cref{eq:GramApproximation}
satisfy $\tau \le 2$, $\gamma \le \frac {\alpha}{ (32/c+2\tau)c' n\mu}$.
If the algorithm terminates at line 3, then the returned variables $\lambda_{ij}$ and matrix $N$
are valid output of the reduced oracle. Matrix $N$ has at most $O(m)$
non-zero entries, and its spectral norm is at most $2\mu$.
\end{lemma}
\begin{proof}
Lemma~\ref{lem:flow-subroutine-routing}(b) gives the capacity bounds on $\lambda$, the identity
$N=\sum_p f_pT_p-\calL(G^\lambda)=-\calL(D)$, and $\|N\|\le2\mu$.
It remains to show that $N\bullet X\le -\alpha$. 
By Assumption~\ref{assumption:main}, $\|x\|^2\le 4/c$ for each $x\in S$, and hence $\|x-y\|^2\le 16/c$ for each $x,y\in S$.
Since $(v_1,\ldots,v_n)\approx_{\gamma,\tau}(\tilde v_1,\ldots,\tilde v_n)$, we get
$\|x-y\|^2 - \|v_x-v_y\|^2 \le \gamma(\|x-y\|^2+\tau)\le \frac{\alpha}{(32/c+2\tau)c'n\mu}(16/c+\tau)=\frac {\alpha}{ 2c' n\mu}$. % this was slightly wrong as-written before
We can thus write
\begin{align*}
-N\bullet X&=\sum_{x\in A,y\in B} d_{xy}\|v_x-v_y\|^2 \ge
\sum_{x\in A,y\in B} d_{xy}(\|x-y\|^2 - \tfrac{\alpha}{2c' n \mu}) \\
&\ge \sum_{x\in A,y\in B} d_{xy}\|x-y\|^2 - \tfrac{\alpha}{2c' n \mu}\sum_{x\in A}\sum_{y\in B}d_{xy}
\ge 2\alpha -  \tfrac{\alpha}{2c' n \mu}\cdot |A| \cdot \mu=\alpha
\end{align*}

\end{proof}

We will assume that if ${\tt Matching}(u)$ terminates at line 2 or 3 then it returns an empty matching. 
Thus, we always have ${\tt Matching}(u) \subseteq V\times V$ and
$|{\tt Matching}(u)|\le |V|$.

\begin{lemma}\label{lemma:matching}
Suppose that $\mu\ge \frac{6 \alpha }{c' n\Delta}$.
There exist positive constants $c',\sigma,\delta$ for which  either (i) $\mathbb E_{u} |{\tt Matching}(u)|\ge \delta n$,
or (ii)~\cref{alg:Matching} for $u\sim\calN(0,I_d)$ terminates at line 2 or 3 with probability at least $\Theta(1)$.
\end{lemma}
\begin{proof}
Choose $c',\sigma,\zeta$ as in Lemma~\ref{lem:flow-subroutine-gaussian}.
If the probability of termination at line 2 or 3 is at least $\zeta/2$, then (ii) holds.
Otherwise, with probability at least $\zeta/2$, the projection gap is at least $\sigma$ and the algorithm reaches line 5.
On this event, let $A'\subseteq A$ and $B'\subseteq B$ be the endpoints of the matching $M$.
The total demand is at least $c'n\mu$, and each vertex is incident to demand at most $\mu$ by Lemma~\ref{lem:flow-subroutine-routing}(b).
Thus, if $|M|<c'n/3$, then
\[
 \sum_{x\in A-A',\,y\in B-B'}d_{xy}
 \ge c'n\mu-\mu(|A'|+|B'|)>\frac{c'n\mu}{3}.
\]
Every positive-demand pair in this sum has projection gap at least $\sigma$ and, by maximality of $M$, squared distance greater than $\Delta$.
Consequently,
\[
 \sum_{x\in A,\,y\in B}d_{xy}\|x-y\|^2
 >\frac{c'n\mu\Delta}{3}\ge2\alpha,
\]
contradicting the fact that line 3 did not terminate.
Hence $|M|\ge c'n/3$ on this event, and $\mathbb E_u|{\tt Matching}(u)|\ge \zeta c'n/6$.
%This proves (i) with $\delta=\zeta c'/6$.
\end{proof}

Let us define $c',\sigma,\delta$ as in \cref{lemma:matching}, and set
\begin{eqnarray}
\Delta&=&\sqrt{\frac\varepsilon{\log n}} \\
\mu&\in&[\mu_0,2\mu_0], \quad \mu_0=\frac{6 \alpha }{c' n\Delta}
\end{eqnarray}
The cost of the element cut in line 2 of the algorithm is then at most $ 2c'n\mu=\Theta\left(\alpha\sqrt{({\log n})/{\varepsilon}}\right)$,
which corresponds to an $O\left(\sqrt{({\log n})/{\varepsilon}}\right)$ pseudo-approximation algorithm.

\subsection{Chaining algorithm}

The remaining part is identical to the algorithm in~\cite{vnk:sparsest-cut}.
Assume  that case (i) in Lemma~\ref{lemma:matching} holds (otherwise calling ${\tt Matching}(u)$ for $u\sim\calN(0,I_d)$
will terminate the oracle after $O(1)$ expected calls).
Assume also that ${\tt Matching}(\cdot)$ is {\em skew-symmetric}, i.e.\ ${\tt Matching}(-u)$ is obtained from ${\tt Matching}(u)$
by reversing edge orientations. (This can be easily enforced algorithmically).
The idea is to use procedure ${\tt Matching}(\cdot)$ to find many ``violating paths'', i.e.\ paths $p=(p_0,\ldots,p_{\ell(p)})$ that satisfy
\begin{equation}\label{eq:def:violating}
\sum_{j=1}^{\ell(p)} \|p_j-p_{j-1}\|^2  \;\;\le\;\; \| p_{\ell(p)}-p_0\|^2 - \Delta
\end{equation}

For a set of directed paths $M$ let $M^{\tt violating}$ be the set of paths $p\in M$ that are either violating
or contain a violating path $q$ as a subpath ($p=(\ldots,q,\ldots)$). For two sets of paths $M,M'$ define
\[
M\circ M'=\{p\circ q\::\:p\in M,q\in M',{\tt endpoint}(p)={\tt startpoint}(q)\}
\]

\begin{theorem}[\cite{Sherman,vnk:sparsest-cut}]\label{th:vnk:chaining}
Let ${\tt Matching}(u)$ be a procedure that for given $u\in\mathbb R^d$ returns a directed matching on $V$
in which every edge $(x,y)$ satisfies $\langle y-x,u\rangle\ge \sigma$ and $\|x-y\|^2\le \Delta$.
Assume that it is skew-symmetric and $\mathbb E_u[|{\tt Matching}(u)|]=\Omega(n)$.

\begin{sloppypar}
For any $K\in[1,O(1/\Delta)]$ (with an appropriate constant) 
there exists an efficiently samplable distribution $\calD$ of vectors ${\bf u}=(u_1,\ldots,u_K)$ %with $K=O(1/\Delta)$
such that set $M({\bf u})={\tt Matching}(u_1)\circ\ldots\circ {\tt Matching}(u_K)$
satisfies $\mathbb E_{{\bf u}\sim\calD}|M({\bf u})^{\tt violating}|\ge e^{-\Theta(K^2)}n$.
\end{sloppypar}
\end{theorem}

Set $K=\Theta(\Delta\log n)=\Theta\left(\sqrt{\varepsilon\log n}\right)$ with a sufficiently small constant, then indeed $K=O(1/\Delta)=O\left(\sqrt{(\log n)/\varepsilon}\right)$
since $\varepsilon=O(1)$.
Let us run the algorithm in \cref{th:vnk:chaining}
until we get $|M({\bf u})^{\tt violating}|\ge e^{-\Theta(K^2)}n=n^{1-\Theta(\varepsilon)}$.
Since we always have $|M({\bf u})|\le n$, this will happen after $ n / n^{1-\Theta(\varepsilon)} = n^{\Theta(\varepsilon)}$ runs
in expectation (or $Kn^{\Theta(\varepsilon)}=n^{\Theta(\varepsilon)}$ calls to ${\tt Matching}(\cdot)$).
Once we have a set $P$ of violating paths with $|P|\ge n^{1-\Theta(\varepsilon)}$,
the output of the oracle is constructed as follows.
Set $\lambda_{ij}=0$ for all $ij\in E$. Set $f_p=\frac{2\alpha}{|P|\Delta}$ for all $p\in P$,
and $f_p=0$ for all other paths.
Then $N=\sum\nolimits_p f_p T_p - \calL(G^\lambda)=\frac{2\alpha}{|P|\Delta}(\calL(F)-\calL(D))$
where multigraph $F$ is the union of paths in $P$
and multigraph $D$ is the union of edges $\partial p=\{{\tt startpoint}(p),{\tt endpoint}(p)\}$ over $p\in P$.
Note that the complexity of the oracle is dominated by $n^{\Theta(\varepsilon)}$ maximum flow computations.
\begin{lemma}
\label{lemma:oracle-fin}
Suppose that parameters $\tau,\gamma$ in \cref{eq:GramApproximation}
satisfy  $\tau<2$ and $\gamma\le \tfrac{\Delta}{20(K+1)}$.
Then variables $\lambda_{ij}$ and matrix $N$ defined above
are valid output of the reduced oracle. Matrix $N$ has at most $O(n\sqrt{\log n})$
non-zero entries, and its spectral norm is at most
$\frac{\alpha\cdot\operatorname{polylog}(n)}{n^{1-\Theta(\varepsilon)}}$.
\end{lemma}
\begin{proof}
We need to show that $N\bullet X\le -\alpha$. Using \cref{eq:def:violating} and the fact that
$(v_1,\ldots,v_n)\approx_{\gamma,\tau}(\tilde v_1,\ldots,\tilde v_n)$, we conclude that for any $p\in P$ we have
\begin{equation*}
\sum_{j=1}^{\ell(p)} \|v_{p_j}-v_{p_{j-1}}\|^2  \;\;\le\;\; \| v_{p_{\ell(p)}}-v_{p_0}\|^2 - \tfrac 12\Delta
\qquad\qquad T_p\bullet X\le - \tfrac 12\Delta
\end{equation*}
(using the same argument as in the proof of \cref{lemma:oracle7}).
Therefore, $N\bullet X\le |P|\cdot \frac{2\alpha}{|P|\Delta} \cdot (-\tfrac 12 \Delta)=-\alpha$, as desired.

By construction, matrix $N$ has at most $O(|P|\cdot K)\le O(n\sqrt{\log n})$ non-zero entries.
Graphs $F$ and $D$ have maximum degree at most $2K$, and hence $\|N\|\le \frac{2\alpha}{|P|\Delta} (\|\calL(F)\|+\|\calL(D)\|)\le \frac{2\alpha}{n^{1-\Theta(\varepsilon)}\Delta}(2K+2K)
\le\frac{\alpha\cdot\operatorname{polylog}(n)}{n^{1-\Theta(\varepsilon)}}$.
\end{proof}

\subsection{Wrap-up: proof of the second oracle bound}

We now collect the complexity, width, and sparsity bounds for the second
implementation of the reduced oracle.

\begin{lemma}\label{lemma:oracle-bounds}
Let $\varepsilon\in[\Theta(1/\log n),\Theta(1)]$.  The oracle described above
runs in $O(n^{O(\varepsilon)}\cdot\tflow(m,m)\cdot\operatorname{polylog}(n))$
time and produces either a $\Theta(1)$-balanced element cut of cost
$O(\alpha\sqrt{\log n/\varepsilon})$ or a valid feedback matrix with spectral
norm $O(\alpha\sqrt{\log n}/(n^{1-\Theta(\varepsilon)}\sqrt\varepsilon))$,
with probability at least $1-n^{-\Theta(1)}$.  Moreover, every feedback
matrix has $O(m\sqrt{\log n})$ non-zero entries.
\end{lemma}

\begin{proof}
Choose parameters $\Delta=\sqrt{\varepsilon/\log n}$ and
$\mu\in[\mu_0,2\mu_0]$, where $\mu_0=6\alpha/(c'n\Delta)$, as described
above.  The value $\mu$ can be chosen as a ratio $p/q$ of two
polynomially-bounded integers: since $\mu_0\ge c''/n$ for an appropriate
constant $c''>0$, take
$p=\lceil(2n/c'')\mu_0\rceil$ and $q=\lfloor2n/c''\rfloor$.
Multiplying all finite capacities by $2q$ makes them
polynomially-bounded integers, including the capacities $w(x)/2$; after
computing a flow, scale it back by $2q$.

The oracle repeatedly runs Algorithm~\ref{alg:Matching}.  If it terminates
at line 2, Lemma~\ref{lem:flow-subroutine-routing}(a) gives a
$\Theta(1)$-balanced element cut of cost at most
$2c'n\mu=O(\alpha\sqrt{\log n/\varepsilon})$.  If it terminates at line 3,
Lemma~\ref{lemma:oracle7} gives a valid feedback matrix with spectral norm
$O(\mu)=O(\alpha\sqrt{\log n}/(n\sqrt\varepsilon))$ and $O(m)$ non-zero
entries.  Otherwise, Lemma~\ref{lemma:oracle-fin} gives a valid feedback
matrix with spectral norm
$\alpha n^{-1+\Theta(\varepsilon)}\operatorname{polylog}(n)$ and
$O(n\sqrt{\log n})$ non-zero entries.  Since $m=\Omega(n)$, these alternatives
satisfy the width and sparsity bounds in the statement.

Each call to Algorithm~\ref{alg:Matching} makes one maximum-flow call and
otherwise takes $O(m\operatorname{polylog}(n))$ time.  The chaining
procedure uses $n^{O(\varepsilon)}$ such calls per trial.  A trial succeeds
with probability at least $n^{-O(\varepsilon)}$ by the reverse Markov
inequality, so running $n^{O(\varepsilon)}\log n$ independent trials in
lockstep gives a valid answer with probability at least
$1-n^{-\Theta(1)}$.  The total runtime is therefore
$O(n^{O(\varepsilon)}\cdot\tflow(m,m)\cdot\operatorname{polylog}(n))$.

Taking $\tau=\Theta(1)$ sufficiently small and
$\gamma=\Theta(\Delta/(K+1))=\Theta(1/\log n)$ meets the hypotheses of
Lemmas~\ref{lemma:oracle7} and~\ref{lemma:oracle-fin}.  This also gives the
common parameter bounds in Theorem~\ref{thm:reduced-oracle}.
\end{proof}

\section{Acknowledgments}

LLM assistance was used to derive \cref{thm:main:padding} as a fast algorithm
(particularly, GPT 6-Astra and Opus 5.5).
However, the ideas that underly the
approximation itself were human-devised,
based on the ``edge-wise non-expansive''
line embeddings originally used in
\cite[Section~3]{spalding2025reweighted}.
%%%%%%%%%%%%%%%%%%%%%%%%%%%%%%%%%%%%%%%%%%
%This specific kind of embedding was applied
%in an effort to avoid the pitfalls leading
%to the errors mentioned in the footnotes of
%\cref{tab:classes}.\todo{Should get permission for this as well.}
%%%%%%%%%%%%%%%%%%%%%%%%%%%%%%%%%%%%%%%%%%
LLMs were not used to derive any other results in this work.

All AI-assisted mathematical arguments were verified by the authors, who significantly
revised and simplified them. The authors take full responsibility for the
correctness and content of the manuscript. 
\bibliography{maxflow}

\begin{thebibliography}{10}

\bibitem{agarwal2005}
Amit Agarwal, Moses Charikar, Konstantin Makarychev, and Yury Makarychev.
\newblock {$O(\sqrt{\log n})$} approximation algorithms for min {UnCut}, min
  {2CNF} deletion, and directed cut problems.
\newblock In {\em Proceedings of the Thirty-Seventh Annual ACM Symposium on
  Theory of Computing}, STOC '05, page 573–581, New York, NY, USA, 2005.
  Association for Computing Machinery.
\newblock \href {https://doi.org/10.1145/1060590.1060675}
  {\path{doi:10.1145/1060590.1060675}}.

\bibitem{alon2013matrix}
Noga Alon and Raphael Yuster.
\newblock Matrix sparsification and nested dissection over arbitrary fields.
\newblock {\em Journal of the ACM (JACM)}, 60(4):1--18, 2013.

\bibitem{AK}
S.~Arora and S.~Kale.
\newblock A combinatorial, primal-dual approach to semidefinite programs.
\newblock {\em Journal of the ACM}, 63(2):1--35, 2016.

\bibitem{AKstoc}
Sanjeev Arora and Satyen Kale.
\newblock A combinatorial, primal-dual approach to semidefinite programs.
\newblock In {\em Proceedings of the thirty-ninth annual ACM symposium on
  Theory of computing}, pages 227--236, 2007.

\bibitem{brandt2019approximating}
Sebastian Brandt and Roger Wattenhofer.
\newblock Approximating small balanced vertex separators in almost linear time.
\newblock {\em Algorithmica}, 81(10):4070--4097, 2019.

\bibitem{brinkmann2020practical}
Gunnar Brinkmann.
\newblock A practical algorithm for the computation of the genus.
\newblock {\em arXiv preprint arXiv:2005.08243}, 2020.

\bibitem{bui1992finding}
Thang~Nguyen Bui and Curt Jones.
\newblock Finding good approximate vertex and edge partitions is np-hard.
\newblock {\em Information Processing Letters}, 42(3):153--159, 1992.

\bibitem{ChenGraurSidford2025}
Li~Chen, Andrei Graur, and Aaron Sidford.
\newblock Accelerated approximate optimization of multi-commodity flows on
  directed graphs.
\newblock In {\em Proceedings of the 57th Annual ACM Symposium on Theory of
  Computing}, STOC '25, pages 2203--2212, 2025.
\newblock URL: \url{https://arxiv.org/abs/2503.24373}, \href
  {https://doi.org/10.1145/3717823.3718267}
  {\path{doi:10.1145/3717823.3718267}}.

\bibitem{linear-maxflow}
Li~Chen, Rasmus Kyng, Yang Liu, Richard Peng, Maximilian Probst~Gutenberg, and
  Sushant Sachdeva.
\newblock Maximum flow and minimum-cost flow in almost-linear time.
\newblock {\em Journal of the ACM}, 72(3):1--103, 2025.

\bibitem{cheung2013graph}
Ho~Yee Cheung, Lap~Chi Lau, and Kai~Man Leung.
\newblock Graph connectivities, network coding, and expander graphs.
\newblock {\em SIAM Journal on Computing}, 42(3):733--751, 2013.

\bibitem{ConroyFiltser2025}
Jonathan Conroy and Arnold Filtser.
\newblock How to protect yourself from threatening skeletons: Optimal padded
  decompositions for minor-free graphs, 2025.
\newblock URL: \url{https://arxiv.org/abs/2504.00278}, \href
  {http://arxiv.org/abs/2504.00278} {\path{arXiv:2504.00278}}.

\bibitem{davies2025}
James Davies, Agelos Georgakopoulos, Meike Hatzel, and Rose McCarty.
\newblock Strongly sublinear separators and bounded asymptotic dimension for
  sphere intersection graphs.
\newblock In Oswin Aichholzer and Haitao Wang, editors, {\em 41st International
  Symposium on Computational Geometry (SoCG 2025)}, volume 332 of {\em Leibniz
  International Proceedings in Informatics (LIPIcs)}, pages 36:1--36:16,
  Dagstuhl, Germany, 2025. Schloss Dagstuhl -- Leibniz-Zentrum f{\"u}r
  Informatik.
\newblock \href {https://doi.org/10.4230/LIPIcs.SoCG.2025.36}
  {\path{doi:10.4230/LIPIcs.SoCG.2025.36}}.

\bibitem{DRSW27}
James Davies, Pawe{\l} Rz{\k{a}}{\.z}ewski, Jack Spalding-Jamieson, and Bartosz
  Walczak.
\newblock Revisiting separators in string graphs, region intersection graphs,
  and induced-minor-free graphs.
\newblock In {\em Proceedings of the 2027 Annual ACM-SIAM Symposium on Discrete
  Algorithms (SODA)}, 2027.
\newblock URL: \url{https://arxiv.org/abs/2610.23792}.

\bibitem{delling2011graph}
Daniel Delling, Andrew~V Goldberg, Ilya Razenshteyn, and Renato~F Werneck.
\newblock Graph partitioning with natural cuts.
\newblock In {\em 2011 IEEE International Parallel \& Distributed Processing
  Symposium}, pages 1135--1146. IEEE, 2011.

\bibitem{dong2024fast}
Sally Dong, Gramoz Goranci, Lawrence Li, Sushant Sachdeva, and Guanghao Ye.
\newblock Fast algorithms for separable linear programs.
\newblock In {\em Proceedings of the 2024 Annual ACM-SIAM Symposium on Discrete
  Algorithms (SODA)}, pages 3558--3604. SIAM, 2024.

\bibitem{ducoffe2022diameter}
Guillaume Ducoffe, Michel Habib, and Laurent Viennot.
\newblock Diameter, eccentricities and distance oracle computations on h-minor
  free graphs and graphs of bounded (distance) vapnik--chervonenkis dimension.
\newblock {\em SIAM Journal on Computing}, 51(5):1506--1534, 2022.

\bibitem{elberfeld2025space}
Michael Elberfeld, Frank Kammer, and Johannes Meintrup.
\newblock Space-efficient depth-first search via augmented succinct graph
  encodings.
\newblock {\em arXiv preprint arXiv:2504.19547}, 2025.

\bibitem{FeigeHL08}
Uriel Feige, MohammadTaghi Hajiaghayi, and James~R. Lee.
\newblock Improved approximation algorithms for minimum weight vertex
  separators.
\newblock {\em {SIAM} J. Comput.}, 38(2):629--657, 2008.
\newblock \href {https://doi.org/10.1137/05064299X}
  {\path{doi:10.1137/05064299X}}.

\bibitem{filtser2025optimal}
Arnold Filtser, Tobias Friedrich, Davis Issac, Nikhil Kumar, Hung Le, Nadym
  Mallek, and Ziena Zeif.
\newblock Optimal padded decomposition for bounded treewidth graphs.
\newblock {\em TheoretiCS}, 4, 2025.

\bibitem{fox2023simple}
Emily Fox.
\newblock A simple deterministic near-linear time approximation scheme for
  transshipment with arbitrary positive edge costs.
\newblock {\em arXiv preprint arXiv:2307.07440}, 2023.

\bibitem{gilbert1984separator}
John~R Gilbert, Joan~P Hutchinson, and Robert~Endre Tarjan.
\newblock A separator theorem for graphs of bounded genus.
\newblock {\em Journal of Algorithms}, 5(3):391--407, 1984.

\bibitem{kawarabayashi2010separator}
Ken-ichi Kawarabayashi and Bruce Reed.
\newblock A separator theorem in minor-closed classes.
\newblock In {\em 2010 IEEE 51st Annual Symposium on Foundations of Computer
  Science}, pages 153--162. IEEE, 2010.

\bibitem{kerber2016persistent}
Michael Kerber, Donald~R Sheehy, and Primoz Skraba.
\newblock Persistent homology and nested dissection.
\newblock In {\em Proceedings of the Twenty-Seventh Annual ACM-SIAM Symposium
  on Discrete Algorithms}, pages 1234--1245. SIAM, 2016.

\bibitem{vnk:sparsest-cut}
Vladimir Kolmogorov.
\newblock A simpler and parallelizable {$O(\sqrt{\log n})$}-approximation
  algorithm for sparsest cut.
\newblock {\em Transactions on Algorithms}, 21(4):1--22, 2025.

\bibitem{korhonen2024}
Tuukka Korhonen and Daniel Lokshtanov.
\newblock Induced-minor-free graphs: Separator theorem, subexponential
  algorithms, and improved hardness of recognition.
\newblock In {\em Proceedings of the 2024 Annual ACM-SIAM Symposium on Discrete
  Algorithms (SODA)}, pages 5249--5275, 2024.
\newblock \href {https://doi.org/10.1137/1.9781611977912.188}
  {\path{doi:10.1137/1.9781611977912.188}}.

\bibitem{lahn2021weighted}
Nathaniel~Adam Lahn and Sharath Raghvendra.
\newblock A weighted approach to the maximum cardinality bipartite matching
  problem with applications in geometric settings.
\newblock {\em Journal of Computational Geometry}, 11(2):183--204, 2021.

\bibitem{lau2023cheeger}
Lap~Chi Lau, Kam~Chuen Tung, and Robert Wang.
\newblock Cheeger inequalities for directed graphs and hypergraphs using
  reweighted eigenvalues.
\newblock In {\em Proceedings of the 55th Annual ACM Symposium on Theory of
  Computing}, pages 1834--1847, 2023.

\bibitem{LTW}
Lap~Chi Lau, Kam~Chuen Tung, and Robert Wang.
\newblock "fast algorithms for directed graph partitioning using flows and
  reweighted eigenvalues".
\newblock In {\em Proceedings of the 2024 Annual ACM-SIAM Symposium on Discrete
  Algorithms (SODA)}, pages 591--624. SIAM, 2024.

\bibitem{lee2010genus}
James~R Lee and Anastasios Sidiropoulos.
\newblock Genus and the geometry of the cut graph.
\newblock In {\em Proceedings of the twenty-first annual ACM-SIAM symposium on
  Discrete Algorithms}, pages 193--201. SIAM, 2010.

\bibitem{leighton1988approximate}
Tom Leighton and Satish Rao.
\newblock An approximate max-flow min-cut theorem for uniform multicommodity
  flow problems with applications to approximation algorithms.
\newblock In {\em FOCS}, 1988.

\bibitem{leighton1999multi}
Tom Leighton and Satish Rao.
\newblock Multicommodity max-flow min-cut theorems and their use in designing
  approximation algorithms.
\newblock {\em Journal of the ACM}, 46(6):787–832, November 1999.
\newblock \href {https://doi.org/10.1145/331524.331526}
  {\path{doi:10.1145/331524.331526}}.

\bibitem{lipton1979generalized}
Richard~J Lipton, Donald~J Rose, and Robert~Endre Tarjan.
\newblock Generalized nested dissection.
\newblock {\em SIAM journal on numerical analysis}, 16(2):346--358, 1979.

\bibitem{lipton1979separator}
Richard~J Lipton and Robert~Endre Tarjan.
\newblock A separator theorem for planar graphs.
\newblock {\em SIAM Journal on Applied Mathematics}, 36(2):177--189, 1979.

\bibitem{lipton1980}
Richard~J. Lipton and Robert~Endre Tarjan.
\newblock Applications of a planar separator theorem.
\newblock {\em SIAM Journal on Computing}, 9(3):615--627, 1980.
\newblock \href {https://doi.org/10.1137/0209046} {\path{doi:10.1137/0209046}}.

\bibitem{louis2010cutmatchinggamesdirectedgraphs}
Anand Louis.
\newblock Cut-matching games on directed graphs, 2010.
\newblock URL: \url{https://arxiv.org/abs/1010.1047}, \href
  {http://arxiv.org/abs/1010.1047} {\path{arXiv:1010.1047}}.

\bibitem{matousek2014}
Ji{\v{r}}{\'\i} Matou{\v{s}}ek.
\newblock Near-optimal separators in string graphs.
\newblock {\em Combinatorics, Probability and Computing}, 23(1):135–139,
  2014.
\newblock \href {https://doi.org/10.1017/S0963548313000400}
  {\path{doi:10.1017/S0963548313000400}}.

\bibitem{miller98}
Gary~L. Miller, Shang-Hua Teng, William Thurston, and Stephen~A. Vavasis.
\newblock Geometric separators for finite-element meshes.
\newblock {\em SIAM Journal on Scientific Computing}, 19(2):364--386, 1998.
\newblock \href {https://doi.org/10.1137/S1064827594262613}
  {\path{doi:10.1137/S1064827594262613}}.

\bibitem{mohar1999linear}
Bojan Mohar.
\newblock A linear time algorithm for embedding graphs in an arbitrary surface.
\newblock {\em SIAM Journal on Discrete Mathematics}, 12(1):6--26, 1999.

\bibitem{separatorsbook}
Arnold~L. Rosenberg and Lenwood~S. Heath.
\newblock {\em Graph separators with applications}.
\newblock Frontiers of computer science. Kluwer, 2001.

\bibitem{Sherman}
J.~Sherman.
\newblock Breaking the multicommodity flow barrier for {$O(\sqrt{\log
  n})$}-approximations to sparsest cut.
\newblock In {\em FOCS}, pages 363--372, 2009.

\bibitem{SleatorTarjan}
Daniel~D. Sleator and Robert~Endre Tarjan.
\newblock A data structure for dynamic trees.
\newblock {\em J. of Computer and System Sciences}, 26(3):362--391, 1983.

\bibitem{spalding2025reweighted}
Jack Spalding-Jamieson.
\newblock Reweighted spectral partitioning works: A simple algorithm for vertex
  separators in special graph classes.
\newblock {\em arXiv preprint arXiv:2506.01228}, 2025.

\bibitem{ungar1951theorem}
Peter Ungar.
\newblock A theorem on planar graphs.
\newblock {\em Journal of the London Mathematical Society}, 1(4):256--262,
  1951.

\bibitem{zou2026generalized}
Zongrui Zou, Chenglin Fan, Michael Dinitz, Jingcheng Liu, and Jalaj Upadhyay.
\newblock A generalized binary tree mechanism for private approximation of
  all-pair shortest distances.
\newblock In D.~Belgrave, C.~Zhang, H.~Lin, R.~Pascanu, P.~Koniusz,
  M.~Ghassemi, and N.~Chen, editors, {\em Advances in Neural Information
  Processing Systems}, volume~38, pages 36439--36468. Curran Associates, Inc.,
  2025.

\end{thebibliography}

\end{document}